\documentclass[11pt,USenglish,sort]{article}

\usepackage{chngcntr}
\usepackage{todonotes}

\let\oldparagraph=\paragraph
\renewcommand\paragraph[1]{\oldparagraph{#1.}}

\usepackage[T1]{fontenc}
\usepackage{authblk}
\usepackage{a4wide}
\usepackage[pagebackref,colorlinks,linkcolor=red!50!black,citecolor=green!50!black]{hyperref}
\usepackage{amsthm}
\usepackage{amsmath}
\usepackage{amssymb}
\usepackage{subcaption}
\usepackage{amsfonts}
\usepackage{dsfont}
\usepackage{subcaption}
\usepackage[font=small,labelfont=bf]{caption}
\usepackage{thmtools}
\usepackage{thm-restate}
\usepackage[linesnumbered,boxruled]{algorithm2e}
\usepackage{mathtools}
\usepackage{paralist}
\usepackage[nameinlink]{cleveref}

\usepackage{booktabs}
\usepackage{tabularx}
\usepackage{multirow}
\usepackage{makecell}
\usepackage{enumitem}
\usepackage[sort,numbers]{natbib} %

\usepackage{tikz}
\usetikzlibrary{positioning,backgrounds,patterns,calc,shapes,arrows}

\usepackage{etoolbox}

\newcommand{\tempinterval}{temporal interval graph}
\newcommand{\utempinterval}{temporal unit interval graph}
\newcommand{\outempinterval}{order-preserving \utempinterval}

\newtheorem{theorem}{Theorem}[section]
\newtheorem{lemma}[theorem]{Lemma}
\newtheorem{corollary}[theorem]{Corollary}
\newtheorem{observation}[theorem]{Observation}
\newtheorem{proposition}[theorem]{Proposition}

\usepackage{apptools}

\crefname{rrule}{Reduction Rule}{Reduction Rules}
\crefname{claim}{Claim}{Claims}
\crefname{paragraph}{Paragraph}{Paragraphs}
\crefname{observation}{Observation}{Observations}
\crefname{lemma}{Lemma}{Lemmata}
\crefname{theorem}{Theorem}{Theorems}
\crefname{proposition}{Proposition}{Propositions}
\crefname{corollary}{Corollary}{Corollaries}
\crefname{remark}{Remark}{Remarks}
\crefname{section}{Section}{sections}
\crefname{figure}{Figure}{Figures}
\crefname{table}{Table}{Tables}
\crefname{definition}{Definition}{Definitions}
\crefname{equation}{Equation}{Equations}
\crefname{algorithm}{Algorithm}{Algorithms}

\Crefname{remark}{Remark}{Remarks}
\Crefname{paragraph}{Paragraph}{Paragraphs}
\Crefname{theorem}{Thm.}{Thms}
\Crefname{proposition}{Prop}{Props}
\Crefname{corollary}{Cor}{Cors}
\Crefname{observation}{Obs}{Obs}
\Crefname{lemma}{Lem.}{Lems.}

\theoremstyle{definition}
\newtheorem{definition}[theorem]{Definition}

\theoremstyle{remark}

\theoremstyle{plain}

\newcommand{\problemdef}[3]{
		\begin{center}
	\begin{minipage}{0.95\textwidth}
		\noindent
		\textsc{#1}

				\vspace{2pt}
				\setlength{\tabcolsep}{3pt}
				\renewcommand{\arraystretch}{1.1}
				\begin{tabularx}{\textwidth}{@{}lX@{}}
						\textbf{Input:} 		& #2 \\
						\textbf{Question:} 	& #3
					\end{tabularx}
	\end{minipage}
		\end{center}
}

\usepackage{amsopn}

\newcommand{\FPT}{\textnormal{\textsf{FPT}}}
\newcommand{\Wone}{\textnormal{\textsf{W[1]}}}

\newcommand{\XP}{\textnormal{\textsf{XP}}}
\newcommand{\N}{\mathbb N}

\newcommand{\T}{\mathcal{T}}
\newcommand{\I}{\mathcal{I}}
\newcommand{\TG}{\boldsymbol{G}}
\newcommand{\TE}{\boldsymbol{E}}

\newcommand{\NP}{\textnormal{\textsf{NP}}}

\newcommand{\ON}{\mathcal{O}}

\newcommand{\impl}{\Rightarrow}
\newcommand{\yes}{\textnormal{\textsf{yes}}}
\newcommand{\no}{\textnormal{\textsf{no}}}
\newcommand{\ug}[1]{\ensuremath{#1_\downarrow}}

\newcommand{\nonstrproblem}{\textnormal{\textsc{Temporal $(s,z)$-Sep\-a\-ra\-tion}}}

\newcommand{\strproblem}{\textnormal{\textsc{Strict Temporal $(s,z)$-Sep\-a\-ra\-tion}}}

\newcommand{\sepproblem}{\textnormal{\textsc{$(s,z)$-Sep\-a\-ra\-tion}}}

\newcommand{\vertexcover}{\textnormal{\textsc{Vertex Cover}}}

\newcommand{\nonstrpath}[1]{temporal~$(#1)$-path}

\newcommand{\strpath}[1]{strict temporal~$(#1)$-path}

\newcommand{\npath}[1]{$(#1)$-path}

\newcommand{\nonstrsep}[1]{temporal~$(#1)$-sep\-a\-ra\-tor}

\newcommand{\strsep}[1]{strict temporal~$(#1)$-sep\-a\-ra\-tor}

\newcommand{\nsep}[1]{$(#1)$-sep\-a\-ra\-tor}

\newcommand{\exble}{extendable}
\newcommand{\exsion}{extension}

\newcommand{\implone}{\smallskip\noindent$\Rightarrow$:}
\newcommand{\impltwo}{\smallskip\noindent$\Leftarrow$:}

\usepackage{xparse}

\usepackage{mathtools}
\DeclareMathOperator{\tw}{tw}

\newcommand{\abs}[1]{{\left|#1\right|}}

\DeclareMathOperator{\vc}{vc}

\DeclareMathOperator{\td}{td}

\DeclareMathOperator*{\argmin}{arg\,min}
\DeclareMathOperator*{\argmax}{arg\,max}

\DeclarePairedDelimiterX{\set}[1]{\{ }{ \} }{\setargs{#1}}
\NewDocumentCommand{\setargs}{>{\SplitArgument{1}{;}}m}
{\setargsaux#1}
\NewDocumentCommand{\setargsaux}{mm}
{\IfNoValueTF{#2}{#1} {#1\,\delimsize|\,\mathopen{}#2}}%

\newdimen\longformulasindent

\newcommand{\condRef}[1]{{\hyperref[#1]{(\ref{#1})}}}

\usepackage{mathrsfs}
\usepackage{yfonts}
\newcommand{\thetitle}{Temporal Graph Classes: A View Through Temporal Separators}

  \title{\thetitle{}}

\author{Till~Fluschnik\thanks{Supported by the DFG, project DAMM (NI 369/13) and TORE (NI 369/18).}}
\author{Hendrik~Molter\thanks{Supported by the DFG, project MATE (NI 369/17).}} 
\author{Rolf~Niedermeier}
\author{Malte~Renken\thanks{Supported by the DFG, project MATE (NI 369/17).}}
\author{Philipp~Zschoche}

\affil{\small Algorithmics and Computational Complexity, Faculty~IV, TU Berlin, Berlin, Germany,\\
 \texttt{\{till.fluschnik,h.molter,rolf.niedermeier,m.renken,zschoche\}@tu-berlin.de}}
\date{}

\begin{document}
\maketitle
  \begin{abstract}
    We investigate for temporal graphs the computational complexity of separating two distinct vertices~$s$ and~$z$ by vertex deletion.
    In a temporal graph, the vertex set is fixed but the edges have (discrete) time labels.
    Since the corresponding \nonstrproblem{} problem is \NP{}-complete, it is natural to investigate whether relevant special cases exist that are computationally tractable.
    To this end, we study restrictions of the underlying (static) graph---there we observe polynomial-time solvability in the case of bounded treewidth---as well as restrictions concerning the ``temporal evolution'' along the time steps.
    Systematically studying partially novel concepts in this direction, we identify sharp borders between tractable and intractable cases.

\medskip

  \noindent\emph{Keywords:}
    Temporal Paths, 
    Temporal Restrictions, 
    Unit Interval Graphs, 
    NP-completeness,
    Fixed-Parameter Tractability, 
    Dynamic Programming 

  \end{abstract}

\section{Introduction}
  \label{sec:intro}

Reachability, connectivity, and robustness in networks depend often on time.
For instance, in public transport or human contact networks, available connections or contacts are time-dependent.
To model such time-dependent aspects, one turns from static graphs to temporal graphs.
Formally, an undirected \emph{temporal graph}~$\TG = (V,\TE,\tau)$ is an ordered triple consisting of a set~$V$ of vertices, a set~$\TE \subseteq \binom{V}{2} \times \{1,2,\dots, \tau\}$ of \emph{time-edges}, and a maximal time label~$\tau \in \N$. 
We study the problem of finding a small set of vertices in a temporal graph whose removal disconnects two designated terminals: a classic, polynomial-time solvable problem in (static) graph theory.
\problemdef{\nonstrproblem{}}
{ A temporal graph~$\TG=(V,\TE,\tau)$, two distinct vertices~$s,z\in V$, and~$k \in \N$.}
{Does $\TG$ admit a \nonstrsep{s,z} of size at most~$k$?   }
Herein, a vertex set~$S\subseteq V\setminus\{s,z\}$ is a \emph{\nonstrsep{s,z}} for a given temporal graph $\TG=(V,\TE,\tau)$ with $s,z\in V$ if there is no \nonstrpath{s,z} in~$\TG - S:= (V \setminus S, \set{ (\{v,w\},t) \in \TE ; v,w \in V \setminus S}, \tau)$.
A \emph{temporal~$(s,z)$-path} of 
length~$\ell$ in a temporal graph~$\TG=(V,\TE,\tau)$ is a sequence
$P = ( (\{v_0,v_1\}, t_1), (\{v_1, v_2\}, t_2)$, $\dots, (\{v_{\ell-1}, v_\ell\}, t_{\ell} ) )$  of time-edges in~$\TE$,
where $s=v_0$, $z=v_\ell$, $v_i\neq v_j$ for all $i, j\in \{0, 1, \dots, \ell\}$ with $i\neq j$, and~$t_i \leq t_{i+1}$ for all~$i \in \{1, 2, \dots, \ell -1\}$.\footnote{
In the literature, temporal paths are also known as journeys \cite{DBLP:journals/ijfcs/XuanFJ03}. However, in some work a journey has strictly increasing labels \cite{Akrida2017temporallyconnected,akrida2017temporal,michail2016introduction,mertzios2013temporal}.}
\nonstrproblem{} is~\NP{}-complete~\cite{kempe2000connectivity}.
In this work, we study \nonstrproblem{} on restricted classes of temporal graphs with the goal to identify computationally tractable cases.

So far, in the literature one basically finds two different directions concerning the definition of temporal graph classes.
One direction is to define temporal graph classes through the underlying graph  (that is, essentially, the graph obtained by forgetting about the time labels of the edges)~\cite{erlebach2015temporal,AxiotisF16,zschoche2017computational}.
Herein, one restricts the input temporal graph to have its underlying graph being contained in some specific graph class.
The other direction is to consider properties expressible through temporal aspects~\cite{casteigts2012time,FlocchiniMS13,KuhnLO10,michail2016traveling}.
Such properties are, for instance, each layer being a subgraph of its succeeding layer, or the temporal graph being periodic, that is, having a subsequence of layers which is repeated in the same order for some periods.
In this work, we study \nonstrproblem{} on temporal graph classes from both directions.

\paragraph{Our contributions} 
We show that \nonstrproblem{} remains \NP{}-complete on many restricted temporal graph classes.
\begin{figure}[t!]
  \centering
  \newcommand{\tworowsSec}[2]{#1 #2}
  \newcommand{\disttoSec}[1]{distance~to #1}
  \tikzstyle{boxes}=[draw,thick, rounded corners=3mm,text width=2.9cm,align=center,text opacity=1,fill opacity=1,fill=white,font=\footnotesize]
  \tikzstyle{unk}=[fill=gray!15!white]
  \def\eps{0.11}
  \def\diff{0.18}
  \newcommand*{\TExtractCoordinate}[1]{\path (#1); \pgfgetlastxy{\XCoord}{\YCoord};}%
  \resizebox{\textwidth}{!}
  {
  \begin{tikzpicture}[]
  \matrix (first) [ampersand replacement=\&,row sep=0.1cm,column sep=0.55cm]
  {
	  
	  \& \node[boxes,minimum height=0.6cm,yshift=-6mm] (perfect) {perfect graphs, chordal graphs, cographs, split graphs, (unit) interval graphs, thresholds graphs};
	  \& \node[boxes,minimum height=0.6cm] (planar) {planar graphs~\cite{zschoche2017computational} \\ ($\tau$ unbounded; $k$?)};
	  \& \node[boxes,minimum height=0.6cm] (btw) {graphs of bounded treewidth (\Cref{thm:fpt-tw-tau})};
	  \\
	  
	  \node[boxes,minimum height=0.6cm] (claw) {claw-free graphs};
	  \& 
	  \& \node[boxes,minimum height=0.6cm] (bip) {bipartite graphs~\cite{zschoche2017computational} \\ ($\tau\geq2$; $k$:~W[1]-h.)};
	  \& \node[boxes,minimum height=0.6cm] (outplanar) {series-parallel graphs, outerplanar graphs};
	  \\
	  
	  \node[boxes,minimum height=0.6cm] (line) {line graphs (\Cref{lem:nphardline})\\ ($\tau\geq 11$; $k$:~W[1]-h.)};
	  \& \node[boxes,minimum height=0.6cm] (complete) {complete-but-one graphs (\Cref{lem:reductocliqueminusone})\\ ($\tau\geq4$; $k$:~W[1]-h.)};
	  \& \node[boxes,minimum height=0.6cm] (trees) {trees};
	  \& \node[boxes,minimum height=0.6cm] (cactus) {cacti};
	  \\
  };
	  \draw[thick,-] (complete) to (perfect);
	  \draw[thick,-] (complete) to (claw);
	  \draw[thick,-] (line) to (claw);
	  \draw[thick,-] (trees) to (bip);
	  \draw[thick,-] ($(trees.north)+(11mm,0)$) to (outplanar);
	  \draw[thick] (cactus)-- (outplanar);
	  \draw[thick] (outplanar)-- (btw);
	  \draw[thick,-] ($(outplanar.north west)+(2mm,0)$) to (planar);
	  \draw[thick,-] (bip.north) to node[midway, above,scale=0.7,sloped]{perfect graphs}(perfect);
	  
	  \node[below of=line,yshift=0mm,xshift=0mm] () {\textbf{NP-complete}};
	  \node[below of=trees,yshift=0mm,xshift=22mm,text width=6cm] () {\textbf{polynomial-time solvable}};

	  \begin{pgfonlayer}{background}

		  \draw[rounded corners,draw=red,ultra thick]
			  ($(line.west |-complete.south) 				+ ( -\eps,-\diff-6mm)$)--
			  ($(complete.south east) 				+ ( \eps,-\diff-6mm)$)--
			  ($(complete.east |- complete.north) 				+ ( \eps,\diff)$)--
			  ($(bip.east |- complete.north) 				+ ( \eps,\diff)$)--
			  ($(planar.east |- perfect.north) 				+ ( \eps,\diff)$)--
			  ($(perfect.west |- perfect.north) 				+ ( -\eps,\diff)$)--
			  ($(perfect.west |- claw.north) 				+ ( -\eps,\diff)$)--
			  ($(claw.north west) 				+ ( -\eps,\diff)$)--
			  cycle;
			  
		  \draw[rounded corners,draw=green,ultra thick]
			  ($(trees.north west)	+(-\eps, \diff)$)--
			  ($(outplanar.west |- trees.north)	+(-\eps, \diff)$)--
			  ($(btw.west |- perfect.north)	+(-\eps, \diff)$)--
			  ($(btw.east |- perfect.north)	+(\eps, \diff)$)--
			  ($(cactus.south east |-complete.south)	+(\eps, -\diff-6mm)$)--
			  ($(trees.south west|-complete.south)		+( -\eps,-\diff-6mm)$)--
			  cycle;
			  
	  \end{pgfonlayer}
  \end{tikzpicture}
  }
  \vspace{-6mm}
  \caption{
	  Computational complexity of \nonstrproblem{} for some graph classes of the underlying graph. 
	  An edge between two classes indicates containment of the lower in the upper class.
	  For the classes of line, complete-but-one, bipartite, and planar graphs, we provide for which values of the maximum time label~$\tau$ \NP{}-completeness is proven as well as the parameterized complexity of \nonstrproblem{} when parameterized by the solution size~$k$.
	  Note that in the case of planar graphs our \NP{}-hardness proof only holds for unbounded~$\tau$. Moreover, the parameterized complexity regarding~$k$ is unknown.
  }
  \label{fig:udnerlygcs}
\end{figure}
\begin{compactitem}
 \item \nonstrproblem{} remains \NP{}-complete on temporal graphs whose underlying graph falls into a class of graphs containing complete-but-one graphs (that is, complete graphs where exactly one edge is missing) or line graphs.
	 However, if the underlying graph has bounded treewidth, then \nonstrproblem{} becomes polynomial-time solvable (see~\cref{fig:udnerlygcs} for an overview).
 \item \nonstrproblem{} remains \NP{}-complete on temporal graphs where each layer contains only one edge (\cref{cor:one-edge-layer}).
 In contrast, if we require each layer to be a unit interval graph and impose
 suitable restrictions on how the intervals may change over time, then
 \nonstrproblem{} becomes tractable (\cref{prop:intervalpoly}, \cref{thm:ktdist}).
 \item Regarding temporal graph classes defined solely by restrictions on how the edge
 sets of the layers may change over time, \nonstrproblem{} becomes solvable in
 polynomial time on temporal graphs where one layer contains all others (grounded), on graphs where all layers are identical (1-periodic or 0-steady), or when the number of periods is at least the number of vertices.
 In all other considered cases \nonstrproblem{} remains~\NP{}-complete (see \cref{tab:results} in~\cref{sec:tem-classes} for an overview).
\end{compactitem}

\paragraph{Related work}

\citet{kempe2000connectivity} proved that \nonstrproblem{} is \NP-complete.
\citet{zschoche2017computational} proved that~\nonstrproblem{} remains \NP{}-complete on temporal graphs with bipartite or planar underlying graphs.
Moreover, \nonstrproblem{} is \Wone{}-hard when parameterized by the separator size~$k$~\cite{zschoche2017computational}.

\citet{casteigts2012time} defined twelve different classes of temporal graphs and showed a corresponding inclusion diagram.
Among these classes, they define temporal graph classes with recurrence or periodicity of edges.
On a slightly different notion of the latter class, \citet{FlocchiniMS13}
studied the problem of exploring a temporal graph, that is, asking whether it
is possible to visit all vertices of the graph with a temporal walk.
\citet{KuhnLO10} studied the problem of token dissemination on temporal graphs
where for each time-interval of length~$T$, all layers in the interval admit
the same spanning tree.

The class of temporal graphs with underlying graphs of bounded tree\-width are considered in the context of temporal graph exploration~\cite{erlebach2015temporal} and single-source temporal connectivity~\cite{AxiotisF16}.
\citet{erlebach2015temporal} studied the problem of temporal graph exploration on temporal graphs with underlying graphs being planar and of bounded vertex degree.
They also introduced the class of temporal graphs with regularly present edges, where the number of consecutive time steps for which any edge can be absent is lower- and upper-bounded (a similar class without the lower bound is also introduced by \citet[Class 7]{casteigts2012time}).
\citet{michail2016traveling} studied a temporal version of the \textsc{Traveling Salesperson Problem} on temporal graphs with respect to the smallest number~$d$ 
such that every vertex can reach any other vertex at any time in at most~$d$ time steps.

\paragraph{Organization}
In \cref{sec:prem} we introduce all necessary notation and terminology
concerning graph theory and (parameterized) computational complexity theory. In the next three section, we discuss and investigate three canonical and incomparable ways to restrict temporal graphs: In
\cref{sec:layer} we present our results for \nonstrproblem{} on temporal graph
classes that are defined by restricting the layers to be contained in certain graph classes. In
\cref{sec:ug-classes} we present our results for \nonstrproblem{} on temporal
graphs with restricted underlying graphs. In \cref{sec:tem-classes} we discuss some \emph{temporal} restrictions known from the literature that restrict how the edge sets of layers may relate to each other. 
In \cref{sec:unitinterval} we introduce a new class of temporal graphs that combines restrictions on the layers with temporal restrictions and hence does not fit in any of the previous three categories: (almost)
\outempinterval{}s and we present our results for \nonstrproblem{} on those
temporal graphs. We conclude in \cref{sec:conclusion}.

\section{Preliminaries}\label{sec:prem}
As a convention, $\N$ denotes the natural numbers without zero.
For~$n\in \N$, we use~$[n]:=[1:n]:=\{1,2,\ldots,n\}$.
Analogously, for a sequence $x_1,x_2,\ldots,x_n$ and~$a,b\in[n]$,~$a<b$, we write~$x_{[a:b]}$ for the subsequence $x_a,x_{a+1},\ldots,x_b$.

\paragraph{Static graphs}
We use basic notations from (static) graph theory~\cite{diestel2000graphentheory}.
Let $G = (V, E)$ be an \emph{undirected, simple graph}. 
$V(G)$ and $E(G)$ denote the set of vertices and set of edges of~$G$, respectively.
We denote by~$G - V' :=(V\setminus V',\set{ \{v,w\} \in E; v,w \in V \setminus V' })$ 
the graph~$G$ without the vertices in~$V'\subseteq V$.
For~$V' \subseteq V$,~$G[V']:=G - (V \setminus V')$ is the \emph{induced subgraph} of~$G$ on the vertices~$V'$. %
A \emph{path} of length~$\ell$ is sequence of edges~$P = (\{v_1, v_2\}, \{v_2, v_3\},\dots,\{v_\ell, v_{\ell+1}\})$ where $v_i\neq v_j$ for all $i, j\in[\ell-1]$ with $i\neq j$.
We set~$V(P) := \{v_1,v_2,\ldots,v_{\ell+1}\}$.
Path~$P$ is an \emph{$(s,z)$-path} if~$s=v_1$ and~$z=v_{\ell+1}$. %
A set~$S \subseteq V\setminus \{s,z\}$ of vertices is an \emph{\nsep{s,z}}
in~$G$ if there is no \npath{s,z} in~$G-S$.

A \emph{tree decomposition} of a graph~$G$ is a pair~$\T := (T,(B_i)_{i \in V(T)})$ consisting of a tree~$T$ and a family~$(B_i)_{i \in V(T)}$ of \emph{bags}~$B_i \subseteq V(G)$, such that %
\begin{compactenum}[(i)]
	\item for all vertices~$v \in V(G)$ the set~$B^{-1}(v) := \set{i \in V(T) ; v \in B_i}$ is non-empty and induces a subtree of~$T$, and
	\item for every edge~$e \in E(G)$ there is an~$i \in V(T)$ with~$e \subseteq B_i$.
\end{compactenum}
The \emph{width} of~$\T$ is~$\max \set{|B_i| - 1 ; i \in V(T)}$. 
The \emph{treewidth}~$\tw(G)$ of~$G$ is defined as the minimal width over all tree decompositions of~$G$.

\paragraph{Temporal graphs}
Let~$\TG = (V,\TE,\tau)$ be a temporal graph.
We call the graph $G_i(\TG) = (V, E_i(\TG))$ the \emph{layer}~$i$ of $\TG$
where $E_i(\TG) := \set{\{v, w\} ; (\{v, w\}, i) \in \TE }$. %
The \emph{underlying graph}~$\ug{\TG}$ of $\TG$ is defined as~$\ug{\TG} := (V,\ug{E})$, where~$\ug{E} := \set{ e ; \exists t: (e,t) \in \TE}$.
(We drop~$\TG$ in the notations if it is clear from the context.)
For~$X \subseteq V$ we define the \emph{induced temporal subgraph} of~$\TG$ by~$X$ by~$\TG[X] := (X,\set{ (\{v,w\},t) \in \TE ; v,w \in X },\tau)$.
We say that a temporal graph $\TG$ is \emph{connected} if its underlying graph $\ug{\TG}$ is connected.
Let~$s,z \in V$.
The \emph{departure time} (\emph{arrival time}) of a \nonstrpath{s,z}~$P = ((e_1,t_1),(e_2, t_2),\dots,(e_\ell,t_\ell))$ is~$t_1$ ($t_{\ell}$), the \emph{traversal time} of~$P$ is~$t_{\ell} - t_1$, and the length of~$P$ is~$\ell$.
The vertices \emph{visited} by~$P$ are denoted by~$V(P) := \bigcup_{i=1}^{\ell} e_i$.
Throughout the whole paper we assume that the temporal input graph $\TG$ is connected and that there is no time-edge between $s$ and $z$.
Furthermore, in accordance with \citet{wu2016efficient} we assume that the time-edge set $\TE$ is ordered by ascending labels.\footnote{If this is not the case, then $\TE$ can be sorted by ascending labels with bucketsort or mergesort in~$\ON(\min\{\tau,|\TE|\log|\TE|\})$ time.}
		The \emph{concatenation} of two temporal graphs $\TG_1 = (V,\TE_1,\tau_1)$, $\TG_2 = (V,\TE_2,\tau_2)$ is denoted by $\TG_1 \circ \TG_2 := (V, \TE_1 \cup \set{ (e,t+\tau_1) ; (e,t) \in \TE_2}, \tau_1 + \tau_2)$.
		Furthermore, we define that $\TG_1^1 := \TG_1$ and $\TG_1^x := \TG_1^{x-1} \circ \TG_1$ for all integers $x \geq 2$.

We begin by noting that one can efficiently find \nonstrpath{s,z}s 
by using the \emph{static expansion} of a temporal graph.
Intuitively, the static expansion of a temporal graph $\TG$ 
is a directed graph consisting of the union of the layers of $\TG$ where each layer has its own vertex set, and 
additional edges from one vertex of a layer to the same vertex in the next layer.

\begin{definition}
	For a temporal graph $\TG = (V = \{v_1,\dots,v_{n-2},s,z\},\TE,\tau)$, the \emph{static expansion} of~$(\TG,s,z)$ is the directed graph~$H := (V',A)$ with 
\begin{align*}
V' &:= \set{ s, z } \cup \set{ u_{t,j} ; j\in[n-2] \land t \in \phi(v_j) }\\
A  &:= A' \cup A_{s} \cup A_{z}\cup A_{\rm col}\\
A' &:=\set{ ( u_{i,j},u_{i,j'}), (u_{i,j'}, u_{i,j}) ; (\{ v_j,v_{j'}\},i) \in \TE}\\
A_{s} &:=\set{ (s,u_{i,j}) ; (\{ s,v_j\},i) \in \TE }\\
A_{z} &:=\set{ (u_{i,j},z) ; (\{ v_j, z\},i) \in \TE }\\
A_{\rm col} &:=\set{ (u_{t,j},u_{t',j}) ; (t,t') \in \vec{\phi}(v_j)\land j\in[n-2]} \,,
\intertext{where, for all $v \in \{v_1,v_2,\dots,v_{n-2}\}=V\setminus \{s,z\}$,}
\phi(v) &:= \set{ t ; t \in[\tau], \exists w : (\{v,w\},t) \in \TE} \\
\vec{\phi}(v) &:= \set{ (t,t')\in \phi(v)^2 ; \allowbreak t < t' \land \nexists t'' \in \phi(v) : t< t'' < t' } \,.
\end{align*}
The set $A_{\rm col}$ is referred to as the set of \emph{column-edges} of~$H$.
\end{definition}

\begin{lemma}
	\label{lemma:get-path}
	Given a temporal graph~$\TG = (V,\TE,\tau)$ and two distinct vertices~$s$ and~$z$, a \nonstrpath{s,z} can be computed in $\ON(|\TE|)$ time.
\end{lemma}
\begin{proof}
	Let~$\TG = (V,\TE,\tau)$ be a temporal graph with vertex set $V :={} \allowbreak \{ v_1, v_2, \dots, v_{n-2} \}  \cup \{ s, z \}$
	and let $H$ be the static expansion of~$\TG$.
Observe that each \nonstrpath{s,z} in~$\TG$ has a one-to-one correspondence to some~$(s,z)$-path in~$H$ and that $H$ can be computed in $\ON(|\TE|)$ time \cite{zschoche2017computational}.
Thus we can find a \nonstrpath{s,z} in $\TG$, using a breadth-first search on the static expansion of $(\TG,s,z)$.
This gives an overall running time of $\ON(|\TE|)$.
\end{proof}

\paragraph{Parameterized complexity}
We use standard notation and terminology from parameterized
complexity~\cite{downey2013fundamentals,flum2006parameterized,Nie06,cygan2015parameterized}
and give here a brief overview of the most important concepts.
A \emph{parameterized problem} is a language $L\subseteq \Sigma^* \times \mathbb{N}$, where $\Sigma$ is a finite alphabet. We call the second component
the \emph{parameter} of the problem.
A parameterized problem is in the complexity class \XP{} if there is an
algorithm that solves each instance~$(I,r)$ in~$|I|^{f(r)}$ time, for some
computable function $f$.
It is \emph{fixed-pa\-ram\-e\-ter tractable} (in the complexity class \FPT{})
if there is an algorithm that solves each instance~$(I, r)$ in~$f(r) \cdot |I|^{\ON(1)}$ time,
for some computable function $f$. 
There is the W-hierarchy of complexity classes for parameterized problems,
of which the most basic one is called \Wone{}. 
All parameterized complexity classes discussed here are closed under parameterized reductions, which may run in \FPT-time and additionally set the new parameter to a value that only depends on the old parameter. 
If a parameterized problem is \Wone-hard, then it is (presumably) not in \FPT{}.

\section{Layer-wise Restrictions for Temporal Graphs}\label{sec:layer} 
Two approaches to define temporal graph classes derive from restricting either
(i) each layer or (ii) the underlying graph 
to be contained in some specific graph class.
Notably, these restrictions are both independent of the order of the layers and hence appear to not fully capture the temporal characteristics of a given temporal graph.
This section considers case~(i), i.e.\ restrictions on the layers of a temporal graph.
Restricting the layers to fall into a specific graph class neither captures any temporal aspect of the temporal graph nor the full picture drawn by all layers together.
In fact, we show that such restrictions alone are not helpful: \nonstrproblem{} is already \NP{}-complete when each layer consists of at most one edge.
\begin{lemma}
	\label{lemma:edge-layer}
	There is a polynomial-time many-one reduction that maps any instance $(\TG = (V,\TE,\tau),s,z,k)$ of \nonstrproblem{} to an equivalent instance $(\TG' = (V,\TE',\tau'),$ $s,z,k)$ such that each layer in $\TG'$ has at most one edge and  $\tau' \leq \tau \cdot |V|^4$.
\end{lemma}
\begin{proof}
	Let~$\TG = (V,\TE,\tau)$ be a temporal graph.
	We construct~$\TG' := (V,\TE',\tau')$ by concatenating for each layer $i$ of $\TG$ a temporal graph $\TG_i^{|E_i|}$ such that there is a temporal path in $\TG_i^{|E_i|}$
	if and only if there is a path in layer $i$ of $\TG$.

	For each layer~$i$ of $\TG$ 
	we construct a temporal graph $\TG_i := (V, \TE_i,\tau_i)$ 
	by fixing an arbitrary total order on the edge set $E_i = \set{ e_1,e_2,\dots,e_m}$ 
	of layer~$i$ in $\TG$ and setting the time-edge set of layer $j$ of~$\TG_i$ to be $\set{ (e_j,j) }$.
	Now, we build $\TG' := \TG_1^{|E_1|} \circ \TG_2^{|E_2|} \circ \dots \circ \TG_{\tau}^{|E_\tau|}$, where $|E_i|$ is the number of edges in layer $i$ of~$\TG$ for all~$i \in [\tau]$.
	This is obviously a polynomial-time construction.
	Since, for all~$i \in [\tau]$, $|E_i| \leq |V|^2$ and each $\TG_i$ has $|E_i|$ many layers, we know that~$\tau' \leq \tau \cdot |V|^4$.

	Let $i \in [\tau]$ and $v, w \in V$.
	Observe that $G_i(\TG)$ is the underlying graph of both, $\TG_i$ and~$\TG_i^{|E_i|}$.
	Since every temporal path is also a path in the underlying graph, it is easy to see that for each \nonstrpath{v,w} in $\TG_i^{|E_i|}$ there is a \npath{v,w} in layer~$i$ of~$\TG$ which visits the vertices in the same order.
	We claim that for each \npath{v,w} $P$ of length $\ell$ in layer $i$ of $\TG$ there is a \nonstrpath{v,w} in $\TG_i^{\ell}$ which visits the vertices in the same order.
	Let $V(P) =: \set{ v = v_0,v_1,\dots, v_{\ell+1}=w}$ such that $v_{j}$ is visited before~$v_{j+1}$, for all $j \in [0:\ell]$.
	We prove the claim by induction on $\ell$.
	If $\ell=1$, then we know that there is a time-edge between $v$ and $w$ in $\TG_1$.
	For the induction step we observe that there is a time-edge between $v=v_0$ and $v_1$ in $\TG_i$ and, 
	by the induction hypothesis, there is a \nonstrpath{v_1,w} of length $\ell-1$ in~$\TG_i^{\ell-1}$
	which visits the vertices  in the same order as $P$. 
	Since~$\ell \leq |E_i|$, 
	we have that for each \npath{v,w} in layer $i$ of $\TG$ there is a \nonstrpath{v,w} in $\TG_i^{|E_i|}$
	which visits the vertices in the same order, where $v,w \in V$ and~$i \in [\tau]$.
	If follows that a vertex set $S \subseteq V \setminus \set{s,z}$ is a
	\nonstrsep{s,z} in $\TG$ if and only if $S$ is a \nonstrsep{s,z} in $\TG'$, because in the construction of $\TG'$ we replaced  layer $i$ of $\TG$ with~$\TG_i^{|E_i|}$.
\end{proof}
\noindent\cref{lemma:edge-layer} together with known hardness
reductions for \nonstrproblem~\cite{kempe2000connectivity,zschoche2017computational}
implies the following.
\begin{corollary}
  \label{cor:one-edge-layer}
	\nonstrproblem{} is~\NP{}-complete and \Wone-hard when parameterized by the separator size $k$ even if each layer has at most one~edge.
\end{corollary}
 
Now we consider a scenario in which the temporal graphs have a certain geometric interpretation. 
For example in data sets where vertices are individuals and edges model physical proximity (see e.g.~\cite{fournet2014contact}), 
it is a plausible assumption that the individual layers are disc intersection graphs (assuming the individuals only move in the plane).
We investigate the restriction to (unit) interval graphs, which constitute the one-dimensional equivalent, meant as a starting point for further research.

Next, we introduce \tempinterval{}s. 
We call a temporal graph $\TG=(V,\TE,\tau)$ a \emph{\tempinterval} if every layer~$G_i$ is an interval graph. 
We say that a temporal graph~$\TG=(V,\TE,\tau)$ is a \emph{\utempinterval} if every layer~$G_i$ is a \emph{unit} interval graph.
By \cref{lemma:edge-layer}, \nonstrproblem{} on \utempinterval{} is \NP-complete.
Furthermore the problem remains \NP-complete even if $\tau$ is constant:

\begin{proposition}\label{thm:uinterval-hardness-fixed-tau}
\nonstrproblem{} on \utempinterval{}s is \NP{}-complete for any fixed $\tau \geq 6$.
\end{proposition}
\begin{proof}
\citeauthor{zschoche2017computational} showed in \cite[Thm.~3.1]{zschoche2017computational}  by reduction from \textsc{Vertex Cover}
that \nonstrproblem{} is \NP-complete for fixed $\tau \geq 2$.
We modify that proof to ensure that each layer of the resulting temporal graph is a unit interval graph.
\problemdef{\vertexcover{}}
	{ An undirected graph~$G=(V,E)$ and an integer~$k \in \N$.   }
	{Is there a subset~$V' \subseteq V$ of size at most~$k$ such that for all~$\{v,w\} \in E$ it holds~$\{v,w\} \cap V' \not = \emptyset$? }

The basic idea behind the reduction is
to create a gadget for each vertex such that one can use two types of vertex sets to separate $s$ from $z$ in this gadget: a small one and large one.
Then, for each edge in the \vertexcover{} instance,
we connect the corresponding gadgets in such a way, that at least in one of the gadgets it is necessary to take the large vertex set.
Hence, taking the large vertex set from a gadget into the \nonstrsep{s,z} corresponds to taking the vertex into the vertex cover.

Let~$\I := (G = (V,E),k)$ be a \vertexcover{} instance and~$n := |V|$.
We  construct a \nonstrproblem{} instance~$\I' := (\TG' := (V', \TE', 6),s,z,n+k)$ by setting
\begin{align*}
V' &:= \set*{ x, v, x_v, x_v', x_{vw} ; v, w \in V, x \in \{s, z\} }
\end{align*}
and
\begin{align*}
\TE' :={} &\left(E_\alpha(s) \times \{1\}\right) \cup \left( E_\alpha(z) \times \{6\} \right)\\
&\cup \left( E_\beta(s) \times \{2\} \right) \cup \left( E_\beta(z) \times \{5\} \right)\\
&\cup \left( E_\gamma(s, z) \times \{4\} \right) \cup \left( E_\gamma(z, s) \times \{3\} \right)\\
&\cup \left( E_\delta \times \{3\} \right)
\end{align*}
where we define, for any $x, y \in \{s, z\}$, the following four edge classes
\begin{align*}
E_\alpha(x) &:= \binom{\{x, x_v, x_v' : v \in V\}}{2} \,,\\
E_\beta(x) &:= \bigcup_{v \in V} \binom{\{x_v, x_v', x_{vv}\}}{2} \cup \binom{\{x_{vw} : w \in V\}}{2} \,,\\
E_\gamma(x, y) &:= \bigcup_{v \in V} \left\{\{x_v, x_v'\}, \{v, x_v\}, \{v, x_v'\}, \{v, y_{vv}\} \right\} \,,\\
E_\delta &:= \set*{ \{s_{vw}, z_{wv}\}, \{s_{wv}, z_{vw}\} ; \{v, w\} \in E }
\end{align*} 
(compare also \cref{fig:vc-to-separator}).

\begin{figure}[t]
\includegraphics[width=\textwidth]{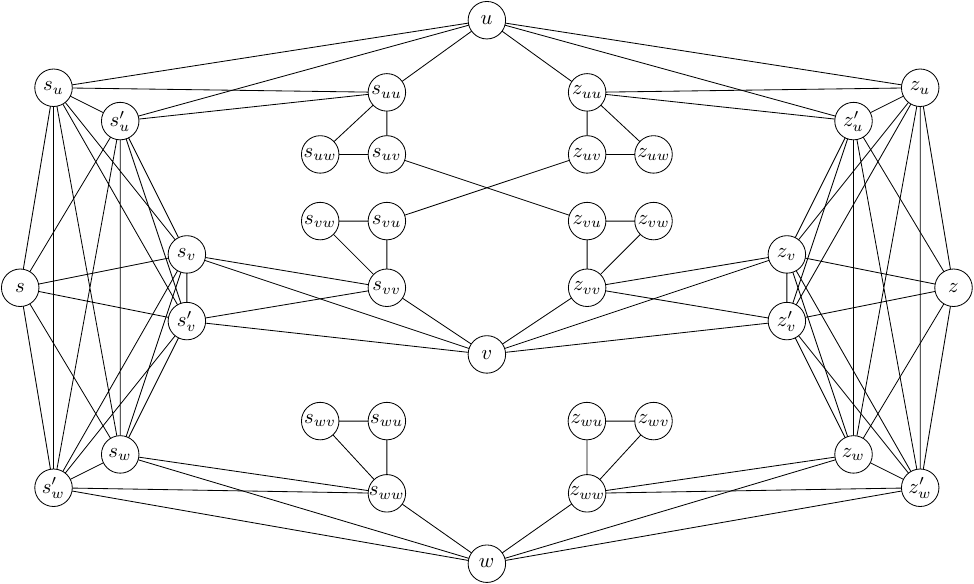}
\caption{Underlying graph of the \nonstrproblem{} instance resulting from a \textsc{Vertex Cover} instance $G = (V, E)$ on three vertices $V = \{u, v, w\}$ and one edge $E = \{\{u, v\}\}$.}
\label{fig:vc-to-separator}
\end{figure}

Observe that no temporal path can use more than one edge from $E_\delta$ as it would need to use an edge from $E_\beta$ in between.
Consequently we may assume that any minimum temporal $(s, z)$-separator only contains vertices from the set $\set{ v, s_{vv}, z_{vv} ; v \in V }$ as we could exchange any other vertex for one of these.
After these observations the rest of the proof works in complete analogy to the proof of \citet[Prop.~3.2]{zschoche2017computational}.

To see that each layer of $\TG'$ is in fact a unit interval graph, first observe that $E_\gamma(z, s)$ and $E_\delta$ are vertex-disjoint and thus each connected component of each layer is taken from a single edge class.
Furthermore, for any choice $x, y \in \{s, z\}$,
\begin{compactitem}
\item $E_\alpha(x)$ forms a clique of size $2n + 1$;
\item each connected component of $E_\beta(x)$ consists of a triangle and a size $n$ clique that share exactly one vertex;
\item each connected component of $E_\gamma(x,y)$ is the union of a triangle and a single edge, joined on a common vertex;
\item $E_\delta$ is a disjoint union of edges.
\end{compactitem}
In summary, each connected component of each layer is either a clique or a union of two cliques sharing a single vertex and thus an interval graph.
\end{proof}

 \section{Restrictions of the Underlying Graph}\label{sec:ug-classes} 
 
After having investigated layer-wise restrictions, we now turn to case (ii),
i.e.\ the study of temporal graphs whose underlying graph is contained in some graph class.
See \cref{fig:udnerlygcs} for an overview of the results.
 
One such class is that of \emph{complete-but-one} graphs, in which all but one possible edges are present.
We show that \nonstrproblem{} is \NP-hard even if the underlying graph of the
temporal input graph is complete-but-one. The main idea that we can reduce the
general problem to that on temporal graphs with a complete-but-one underlying graph by saturating the instance with
``useless'' edges, that do not create any new \nonstrpath{s,z}s.
 \begin{proposition}
  \label{lem:reductocliqueminusone}
  There is a polynomial-time many-one reduction that maps any instance $(\TG = (V,\TE,\tau),s,z,k)$ of \nonstrproblem{} to an equivalent instance~$(\TG' = (V,\TE',\tau'),s,z,k)$ such that~$E(\ug{\TG'})=\binom{V}{2}\setminus\{s,t\}$.
 \end{proposition}
 \begin{proof}
  We construct~$\TG'$ as~$(V,\TE',\tau+2)$ where
\begin{align*}
\TE' :={} &\set*{ (e, t+1) ; (e,t)\in \TE } \\
&\cup \set*{ (\{v, w\}, 1) ; \{v, w \} \in \binom{V \setminus \{s\}}{2} \setminus E(\ug{\TG}) } \\
&\cup \set*{ (\{s, v\}, \tau + 2) ; v \in V \setminus \{z\} \land \{s, v\} \notin E(\ug{\TG}) } \,.
\end{align*}
The one-to-one correspondence of the \nonstrsep{s,z}s in~$\TG$ and~$\TG'$ is immediate.
 \end{proof}
 \noindent\cref{lem:reductocliqueminusone} implies that \nonstrproblem{} remains \NP-complete 
 on all temporal graphs 
 where the underlying graph falls into a graph class containing all complete-but-one graphs, 
 for instance the classes of unit interval or threshold graphs (see \citet{brandstadt1999graph} for definitions).
 We refer to~\cref{fig:udnerlygcs} in~\cref{sec:intro} for an overview.

 Note that complete-but-one graphs are 
 not line graphs (see \citet{brandstadt1999graph} for line graphs), 
 as each complete-but-one graph (with at least five vertices) 
 contains the complete-but-one graph on five vertices 
 as an induced subgraph (see~\citet[Graph~$G_3$]{BEINEKE1970129}).
 Hence, we next study \nonstrproblem{} on temporal graphs where the underlying graph is a line graph.

 \begin{proposition}
  \label{lem:nphardline}
   \nonstrproblem{}  on temporal graphs where the underlying graph is a line graph is \NP{}-complete.
 \end{proposition}

\begin{proof}
	A \nonstrpath{s,z} $P = \big( (\{ s=v_0,v_1 \},t_1), (\{v_1, v_2\}, t_2), \dots, \linebreak ( \{ v_{\ell-1},v_\ell=z \},t_\ell)\big)$ 
	is called \emph{strict} if $t_i < t_{i+1}$ for all $i \in [\ell -1]$.
	A vertex set~$S$ is a \emph{\strsep{s,z}} if there is no \strpath{s,z} in the temporal graph~$\TG - S$.
	The \strproblem{} problem is the ``strict'' variant of \nonstrproblem{} and asks for a \strsep{s,z} instead of a \nonstrsep{s,z}.

	We reduce from the~\NP{}-complete \strproblem{} where each layer is equal and there is no vertex in the underlying graph of degree at most one \cite{zschoche2017computational}.
 Our reduction is similar to the reduction from \strproblem{} to \nonstrproblem{} due to~\citet{zschoche2017computational}.
 Let~$(\TG=(V,\TE,\tau),s,z,k)$ be an instance of~\strproblem{} with~$G_i(\TG)=G_j(\TG)$ for all~$i,j\in[\tau]$.
 We construct an instance $(\TG'=(V',\TE',\tau'),s^*,z^*,k)$ of \nonstrproblem{}, where~$\ug{\TG'}$ is a line graph, as follows.

 \begin{figure}[t!]
 \centering
  \begin{tikzpicture}
    \usetikzlibrary{shapes}
    \tikzstyle{xnode}=[circle, scale=2/3,draw]
    \tikzstyle{xxnode}=[fill, scale=1/2,draw]
    \tikzstyle{lnode}=[fill=green!80!black, color=green!80!black, diamond, scale=1/2.5, draw]
    \tikzstyle{ledge}=[ultra thick,color=green, densely dotted]
    \def\xr{0.7}
    \def\yr{0.7}

    \begin{scope}[rotate=45,scale=0.5]
    \node (a) at (0,0)[xnode]{};
    \node (b) at (2*\xr,0)[xnode]{};
    \node (c) at (0,2*\yr)[xnode]{};
    \node (d) at (-2*\xr,0)[xnode]{};
    \node (e) at (0,-2*\yr)[xnode]{};

    \draw (a) -- (b);
    \draw (a) -- (c);
    \draw (a) -- (d);
    \draw (a) -- (e);

    \end{scope}

    \begin{scope}[xshift=6*\xr cm,rotate=45]
    \node (a) at (0.5*\yr,0.5*\yr)[xnode]{};
    \node (ab) at (0.5*\xr,0)[xxnode]{};
    \node (ac) at (0,0.5*\xr)[xxnode]{};
    \node (ad) at (-0.5*\xr,0)[xxnode]{};
    \node (ae) at (0,-0.5*\xr)[xxnode]{};

    \node (ab1) at (1.5*\xr,-0.5*\yr)[xxnode]{};
    \node (ab2) at (1.5*\xr,0.5*\yr)[xxnode]{};
    \node (ab3) at (2.5*\xr,-0.5*\yr)[xxnode]{};
    \node (ab4) at (2.5*\xr,0.5*\yr)[xxnode]{};

    \node (ac1) at (-0.5*\yr,1.5*\yr)[xxnode]{};
    \node (ac2) at (0.5*\yr,1.5*\yr)[xxnode]{};
    \node (ac3) at (-0.5*\yr,2.5*\yr)[xxnode]{};
    \node (ac4) at (0.5*\yr,2.5*\yr)[xxnode]{};

    \node (ad1) at (-1.5*\xr,-0.5*\yr)[xxnode]{};
    \node (ad2) at (-1.5*\xr,0.5*\yr)[xxnode]{};
    \node (ad3) at (-2.5*\xr,-0.5*\yr)[xxnode]{};
    \node (ad4) at (-2.5*\xr,0.5*\yr)[xxnode]{};

    \node (ae1) at (-0.5*\yr,-1.5*\yr)[xxnode]{};
    \node (ae2) at (0.5*\yr,-1.5*\yr)[xxnode]{};
    \node (ae3) at (-0.5*\yr,-2.5*\yr)[xxnode]{};
    \node (ae4) at (0.5*\yr,-2.5*\yr)[xxnode]{};

    \node (b) at (4*\xr,0)[xnode]{};
    \node (b0a) at (3.5*\xr,0)[xxnode]{};
    \node (c) at (0,4*\yr)[xnode]{};
    \node (c0a) at (0,3.5*\yr)[xxnode]{};
    \node (d) at (-4*\xr,0)[xnode]{};
    \node (d0a) at (-3.5*\xr,0)[xxnode]{};
    \node (e) at (0,-4*\yr)[xnode]{};
    \node (e0a) at (0,-3.5*\yr)[xxnode]{};

    \foreach \x in {b,c,d,e}{
    \draw (a) -- (a\x) -- (a\x1) -- (a\x3) -- (\x0a) -- (a\x4) -- (a\x2) -- (a\x);
    \draw[color=red] (a\x1) -- (a\x2);
    \draw[color=red] (a\x3) -- (a\x4);
    \draw (\x0a) -- (\x);
    }

    \draw[color=red] (ab) -- (ac) -- (ad) -- (ae) -- (ab);
    \draw[color=red] (ab) -- (ad);
    \draw[color=red] (ac) -- (ae);

    \end{scope}

    \begin{scope}[xshift=14*\xr cm,rotate=45]
    \node (a) at (0,0)[lnode]{};
    \node (a0star) at (1,1)[lnode]{};
    \draw[ledge] (a) -- (a0star);

    \node (ab1) at (1*\xr,0*\yr)[lnode]{};
    \node (ab2) at (2*\xr,1*\yr)[lnode]{};
    \node (ab3) at (2*\xr,-1*\yr)[lnode]{};
    \node (ab4) at (3*\xr,0*\yr)[lnode]{};

    \node (ac1) at (0*\yr,1*\yr)[lnode]{};
    \node (ac2) at (1*\yr,2*\yr)[lnode]{};
    \node (ac3) at (-1*\yr,2*\yr)[lnode]{};
    \node (ac4) at (0*\yr,3*\yr)[lnode]{};

    \node (ad1) at (-1*\xr,0*\yr)[lnode]{};
    \node (ad2) at (-2*\xr,1*\yr)[lnode]{};
    \node (ad3) at (-2*\xr,-1*\yr)[lnode]{};
    \node (ad4) at (-3*\xr,0*\yr)[lnode]{};

    \node (ae1) at (0*\yr,-1*\yr)[lnode]{};
    \node (ae2) at (1*\yr,-2*\yr)[lnode]{};
    \node (ae3) at (-1*\yr,-2*\yr)[lnode]{};
    \node (ae4) at (0*\yr,-3*\yr)[lnode]{};

    \node (b) at (3.7*\xr,0)[lnode]{};
    \node (b0star) at (4.3*\xr,0)[lnode]{};
    \node (c) at (0,3.7*\yr)[lnode]{};
    \node (c0star) at (0,4.3*\xr)[lnode]{};
    \node (d) at (-3.7*\xr,0)[lnode]{};
    \node (d0star) at (-4.3*\xr,0)[lnode]{};
    \node (e) at (0,-3.7*\yr)[lnode]{};
    \node (e0star) at (0,-4.3*\xr)[lnode]{};
    \foreach \x in {b,c,d,e}{
    \draw[ledge] (\x) -- (\x0star);
    \draw[ledge]  (a\x1) -- (a\x3) -- (a\x4) -- (a\x2) -- (a\x1);
    \draw[ledge] (a) to (a\x1);
    \draw[ledge] (a\x4) to (\x);
    }
    
    \end{scope}

    \begin{scope}[xshift=14*\xr cm,rotate=45]
    \node (a) at (0.5*\yr,0.5*\yr)[xnode]{};
    \node (ab) at (0.5*\xr,0)[xxnode]{};
    \node (ac) at (0,0.5*\xr)[xxnode]{};
    \node (ad) at (-0.5*\xr,0)[xxnode]{};
    \node (ae) at (0,-0.5*\xr)[xxnode]{};

    \node (ab1) at (1.5*\xr,-0.5*\yr)[xxnode]{};
    \node (ab2) at (1.5*\xr,0.5*\yr)[xxnode]{};
    \node (ab3) at (2.5*\xr,-0.5*\yr)[xxnode]{};
    \node (ab4) at (2.5*\xr,0.5*\yr)[xxnode]{};

    \node (ac1) at (-0.5*\yr,1.5*\yr)[xxnode]{};
    \node (ac2) at (0.5*\yr,1.5*\yr)[xxnode]{};
    \node (ac3) at (-0.5*\yr,2.5*\yr)[xxnode]{};
    \node (ac4) at (0.5*\yr,2.5*\yr)[xxnode]{};

    \node (ad1) at (-1.5*\xr,-0.5*\yr)[xxnode]{};
    \node (ad2) at (-1.5*\xr,0.5*\yr)[xxnode]{};
    \node (ad3) at (-2.5*\xr,-0.5*\yr)[xxnode]{};
    \node (ad4) at (-2.5*\xr,0.5*\yr)[xxnode]{};

    \node (ae1) at (-0.5*\yr,-1.5*\yr)[xxnode]{};
    \node (ae2) at (0.5*\yr,-1.5*\yr)[xxnode]{};
    \node (ae3) at (-0.5*\yr,-2.5*\yr)[xxnode]{};
    \node (ae4) at (0.5*\yr,-2.5*\yr)[xxnode]{};

    \node (b) at (4*\xr,0)[xnode]{};
    \node (b0a) at (3.5*\xr,0)[xxnode]{};
    \node (c) at (0,4*\yr)[xnode]{};
    \node (c0a) at (0,3.5*\yr)[xxnode]{};
    \node (d) at (-4*\xr,0)[xnode]{};
    \node (d0a) at (-3.5*\xr,0)[xxnode]{};
    \node (e) at (0,-4*\yr)[xnode]{};
    \node (e0a) at (0,-3.5*\yr)[xxnode]{};

    \foreach \x in {b,c,d,e}{
    \draw (a) -- (a\x) -- (a\x1) -- (a\x3) -- (\x0a) -- (a\x4) -- (a\x2) -- (a\x);
    \draw[color=black] (a\x1) -- (a\x2);
    \draw[color=black] (a\x3) -- (a\x4);
    \draw (\x0a) -- (\x);
    }

    \draw[color=black] (ab) -- (ac) -- (ad) -- (ae) -- (ab);
    \draw[color=black] (ab) -- (ad);
    \draw[color=black] (ac) -- (ae);
    \end{scope}
  \end{tikzpicture}
  \caption{
	  The underlying graph $\ug{\TG}$ on the left-hand side, the graph $G'$ in the middle, and the graph $H$ (dotted/green) on the right-hand side.
      Red edges (stilts) are the only edges present in layer 1.
    }
  \label{fig:linegraphs}
 \end{figure}
 Let~$G=(V,E):=\ug{\TG}$.
 We construct a graph~$G'=(V',E')$ which will be the underlying graph of~$\TG'$ (refer to~\cref{fig:linegraphs} for an illustration).
 Let~$G'$ be initially a copy of~$G$.
 As a first step, iteratively replace each vertex~$v \in V(G)$ by a set~$W_v$ of~$\deg(v)+1$ vertices such that each edge incident with~$v$ is incident with exactly one vertex from~$W_v$ and every vertex in~$W_v$ is of degree at most one, where $\deg(v)$ denotes the degree of~$v$.
 Note that there is exactly one vertex in~$W_v$ not being incident with an edge, and we call this vertex~$v^*$.
 Denote the edge set of~$G'$ after the first step by~$E''$.
 Next, replace each edge~$\{x,y\}\in E'$ by two paths of length three.
 Denote by $e_{(x,y)}^x$,~$e_{(x,y)}^y$ and by~$e_{(y,x)}^x$,~$e_{(y,x)}^y$ the inner vertices of each path respectively, where $e_{(x,y)}^x$,~$e_{(y,x)}^x$ are neighbors of~$x$ and $e_{(x,y)}^y$,~$e_{(y,x)}^y$ are neighbors of~$y$.
 Next, connect the neighbors of~$x$ on both paths by an edge, and connect the neighbors of~$y$ on both paths by an edge (we refer to these edges as \emph{path stilts}).
 Finally, for each~$v\in V$, turn~$W_v$ into a clique (and refer to all edges in the clique not incident with~$v^*$ as \emph{clique stilts}).
 This finishes the construction of~$G'$.
 It is not hard to see that~$G'$ is indeed a line graph (see \cref{fig:linegraphs} for the graph~$H$ for which holds~$\mathcal{L}(H)=G'$).
 
 We construct~$\TG'$ with vertex set~$V'$ and underlying graph~$G'$ as follows.
 Add the set~$\{(e,1)\mid e\in E'\text{ is a stilt}\}$.
 For each~$2\leq t\leq 2\tau+1$, add the set~$\{(\{v^*,w\},t)\mid w^*\in W_v\setminus\{v^*\}\}$.
 For each~$1\leq t\leq \tau$, add the set of temporal edges $\{(\{x,e_{(x,y)}^x\},2t),(\{e_{(x,y)}^x,e_{(x,y)}^y\},2t),(\{y,e_{(y,x)}^y\},2t)$, $(\{e_{(y,x)}^x,e_{(y,x)}^y\},2t)\mid \{x,y\}\in E''\}$ and
 $\{(\{x,e_{(y,x)}^x\},2t+1),(\{y,e_{(y,x)}^y\},2t+1)\mid \{x,y\}\in E''\}$.
 This finishes the construction of~$\TG'$.
 It is not difficult to see that~$\ug{\TG'}=G'$.
 
 For the correctness, it is enough to observe the following.
 There is no temporal~$(s^*,z^*)$-path starting at time step one.
 It holds that~$\{v,w\}\in E$ if and only if there is a temporal~$(v^*,w^*)$-path starting at~$t$ and ending at~$t+1$ for every~$2\leq t\leq 2\tau$ that does not contain any~$u^*$ except for~$v^*$ and $w^*$.
 We can assume a minimum temporal $(s^*,z^*)$-separator in~$\TG'$ to only contain vertices in~$\{v^*\mid v\in V\}$.
 Hence, the following is immediate: if~$S\subseteq V$ is a strict temporal $(s,z)$-separator in~$\TG$, then~$\{v^*\mid v\in S\}$ is a temporal~$(s^*,z^*)$-separator in~$\TG'$, and vice versa.
\end{proof}
An alternative way to classify an instance of a graph-theoretic problem is through its (graph) parameters.
We study \nonstrproblem{} according to some parameterizations. In the following we show that any upper bound on the maximum length of a \nonstrpath{s,z} leads to a straightforward search-tree algorithm. This gives us a tool to solve \nonstrproblem{} on temporal graphs where the underlying graph is restricted in a way that allows us to upper-bound the length of any temporal path.

\begin{lemma}
	\label{thm:fpt-k-length}
	\nonstrproblem{} is solvable in~$\ON(\ell^{k} \cdot |\TE|)$~time, and thus is fixed-parameter tractable when parameterized by $k+\ell$, 
	where~$k$ is the solution size and~$\ell$ is the maximum length of a \nonstrpath{s,z}.
\end{lemma}
  \begin{proof}
	  We present a depth-first search algorithm (see \cref{alg:strsep}) to show fixed-pa\-ram\-e\-ter tractability.
	  Let~$\I := (\TG = (V,\TE,\tau),s,z,k)$ be an instance of \nonstrproblem{}.
	  The basic idea of our algorithm is simple: at least one vertex of each \nonstrpath{s,z} must be in the \nonstrsep{s,z}.
	  Thus, we compute an arbitrary \nonstrpath{s,z} (Line 4) and branch over all
	  visited vertices of that \nonstrpath{s,z} (Line 9) until we cannot find a
	  \nonstrpath{s,z} in~$\TG-S$, in which case the algorithm outputs \yes, or
	  until we already picked~$k$~vertices to be in the \nonstrsep{s,z}, in which case the algorithm outputs \no.
	  Hence, if the algorithm outputs~\yes, then~$S$ is a \nonstrsep{s,z}.

	  It remains to show that if there is a \nonstrsep{s, z} in~$\TG$, then the algorithm outputs~\yes. 
	  We call a tuple $(S', k')$ a \emph{partial solution} if there is a \nonstrsep{s, z} $S$ of size~$k$ such that $S'\subseteq S$ and $k'\ge k-|S'|$. 
	  Note that $(\emptyset, k)$ is a trivial partial solution. 
	  Now assume \texttt{getSeparator} is called with a partial solution $(S', k')$, then we have that either $S'$ is already a \nonstrsep{s, z} in which case the algorithm outputs~\yes, or there is a \nonstrpath{s, z} $P$ in $\TG - S'$ and a \nonstrsep{s, z}~$S$ such that $S'\subseteq S$. It is clear that $S\cap V(P)\neq \emptyset$, let $v\in S\cap V(P)$. 
	  At some point the algorithm chooses the vertex~$v$ in the for-loop in Line 9 and thus invokes a recursive call with~$(S'\cup \{v\}, k'-1)$. 
	  It is clear that $(S'\cup\{v\})\subseteq S$, we additionally have that $k'-1\ge k-|S'\cup\{v\}|$ since $v\notin S'$. 
	  Hence, we have that $(S'\cup \{v\}, k'-1)$ is a partial solution. 
	  Furthermore, we have that~$|S'| < |S'\cup\{v\}|$. 
	  It is easy to see that if there is a partial solution $(S^\star, k^\star)$ with $|S^\star|=k$, then~$S^\star$ is a \nonstrsep{s, z}. 
	  This implies that the algorithm eventually finds a \nonstrsep{s, z} if one exists and hence is correct.

	  From \cref{lemma:get-path}, we know that we can perform the computation in Line 4 in~$\ON(|\TE|)$ time.
	  Now, we upper-bound the size of the search tree in which each node is a call of the \texttt{getSeparator()} function.
	  We can upper-bound the maximum depth of the search tree by~$k$ as in each recursive call we decrease~$k$ by one, until~$k=0$.
	  Furthermore, a \nonstrpath{s,z} of length at most~$\ell$ visits at most $\ell-1$ vertices different from $s$ and $z$.
	  Thus we can upper-bound the running time of \cref{alg:strsep} by~$\ON(\ell^k \cdot |\TE|)$.
  \end{proof}

  \SetKwProg{Fn}{function}{}{}
  \begin{algorithm}[t]
  \KwIn{A temporal graph~$\TG=(V,\TE,\tau)$, two distinct vertices~$s,z\in V$, and an integer~$k \in \N$.}
  \KwOut{Whether~$\TG$ admits a \nonstrsep{s,z} of size at most~$k$.}

  {\tt getSeparator($\emptyset$,$k$)}\;
	  output \no\;

  \Fn{\tt getSeparator($S$,$k$)}{
	  compute \nonstrpath{s,z}~$P$ in~$\TG-S$\;
	  \uIf{there is no \nonstrpath{s,z} in~$\TG-S$}{
		  output \yes\;
		  exit\;
	  }
	  \ElseIf{$k>0$}{
		  \For{$v \in V(P) \setminus \{s,z\}$}{
			  {\tt getSeparator($S \cup \{v\}$,$k-1$)}\;
		  }
	  }
  }
  \caption{The algorithm behind \cref{thm:fpt-k-length}.}
	  \label{alg:strsep}
  \end{algorithm}

From \cref{thm:fpt-k-length} we can derive that \nonstrproblem{} is linear-time solvable on temporal graph classes where the underlying graph has a constant vertex cover number.\footnote{
The vertex cover number of a graph is the size of the smallest vertex subset that intersects all edges of the graph.}
\begin{corollary}
	\label{thm:fpt-vc}
	\nonstrproblem{} can be solved in~$\ON((2\cdot\vc)^{\vc} \cdot |\TE|)$ time, and thus is fixed-parameter tractable when parameterized by the vertex cover number~$\vc$ of the underlying graph.
\end{corollary}
  \begin{proof}
	  Let $\I := (\TG = (V,\TE,\tau),s,z,k)$ be an instance of \nonstrproblem{} and $\vc$ be the vertex cover number of the underlying graph.
	  We prove this in two steps. We first show that the maximum length $\ell$ of a
	  \nonstrpath{s,z} is upper-bounded by $2\cdot\vc$, and then we show that $k$
	  can be upper-bounded by $\vc$.
	  
	  Since at least one endpoint of each edge of the
	  underlying graph~$\ug{\TG}=(V,\ug{E})$ must be in the vertex cover, the maximum length of a path in~$\ug{\TG}$, and hence the maximum length of a \nonstrpath{s,z}, is at most~$2\cdot\vc$.

	  Without loss of generality we assume that there is no \nonstrpath{s,z}~$P$ of length two, 
	  because the vertex~$v \in V(P) \setminus \set{s,z}$ must be contained in every~\nonstrsep{s,z}. 
	  We can find such a \nonstrpath{s,z} by restricting the breadth-first search
	  of \cref{lemma:get-path} such that it explores only vertices which are reachable by a path which contains at most two non-column edges in the static expansion.
	  Let $V' \subseteq V$ be a vertex cover of size at most $\vc$ for $\ug{\TG}$.
	  The graph $\ug{\TG} - (V' \setminus \set{s,z})$ does not contain any \npath{s,z}s of length greater than two
	  because all remaining edges are incident with~$s$ or~$z$.
	  By our assumption, we know that neither of these \npath{s,z}s correspond to a \nonstrpath{s,z} in $\TG$.
	  Hence, $k < \vc$ or $\I$ is a \yes-instance.
	  It is well-known that if $\ug{\TG}$ admits a vertex cover of size~$\vc$, then
	  we can compute one in~$\ON(2^{\vc}\cdot |\ug{E}|)$~time~\cite{cygan2015parameterized}.
	  The application of \cref{thm:fpt-k-length} completes the proof.
  \end{proof}
Another graph parameter which upper-bounds the maximum length of an \npath{s,z} in the underlying graph is the \emph{tree-depth} of the underlying graph.
First, we provide a formal definition of tree-depth. 
For more details, we refer to \citet{de2012sparsity}.
	\begin{definition}
		The \emph{tree-depth} for graph $G$ with connected components $G_1,G_2,\dots,G_p$ is recursively defined by:
		\[
		\td(G) :=
		\begin{cases}
			1 & \text{if $G$ has only one vertex,}\\
			\max_{i \in [p]} \td(G_i) & \text{if $G$ is not connected, and}\\
			1 + \min_{v \in V(G)} \td(G - \set{v}) & \text{if $G$ is connected.}
		\end{cases}
		\]
	\end{definition}
\begin{corollary}
	\label{thm:fpt-td-k}
	\nonstrproblem{} is solvable in~$\ON(2^{\td(\ug{\TG}) \cdot k} \cdot |\TE|)$ time, and thus is fixed-parameter tractable when parameterized by $k + \td(\ug{\TG})$,
	where $k$ is the solution size.
\end{corollary}
\begin{proof}[Proof of \cref{thm:fpt-td-k}]
	The tree-depth of a graph $G$ is bounded by~$ \log_2(h) \leq \td(G)$ \cite[Lemma~17.2]{de2012sparsity}, where~$h$ denotes the height of a depth-first search tree of~$G$.
	It follows that $h\leq 2^{\td(G)}$ and hence, all paths in $G$ are of length at
	most $2^{\td(G)}$.
	Then, application of \cref{thm:fpt-k-length} completes the proof.
\end{proof}
One of the tools from the repertoire for designing fixed-parameter algorithms
for (static) graph problems are tree
decompositions~\cite{downey2013fundamentals,flum2006parameterized,Nie06,cygan2015parameterized}.
A tree decomposition is a mapping of a graph into a related tree-like structure.
For many graph problems this tree-like structure can be used to formulate a bottom-up dynamic program that starts at the leaves and ends at the root of the tree decomposition.
Indeed, if we parameterize by the treewidth of the underlying graph~$\tw(\ug{\TG})$, then we obtain an \XP-algorithm by dynamic programming. 
Furhtermore, if we add the maximum label $\tau$ to the parameter, then we obtain fixed-parameter tractability when parameterized by $\tw(\ug{\TG})+\tau$.
\begin{theorem}
	\label{thm:fpt-tw-tau}
	For a given tree decomposition of the underlying graph,
	one can solve \nonstrproblem{} in $\ON((\tau+2)^{\tw(\ug{\TG})+2} \cdot \tw(\ug{\TG}) \cdot |V| \cdot |\TE|)$ time,
	where~$\tau$ is the maximum time label.
\end{theorem}
\begin{figure}[t!]
	\includegraphics[width=\textwidth]{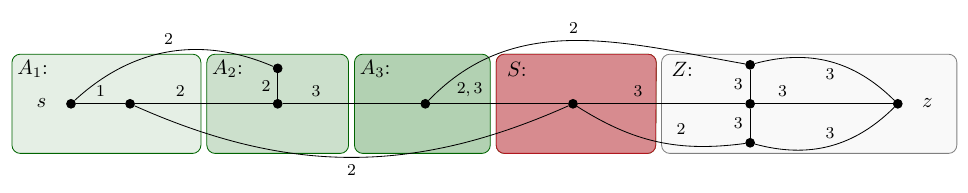}
	\caption{
		The idea for the dynamic program from \cref{thm:fpt-tw-tau} for a temporal graph $\TG$. 
		Vertices in $S$ form the \nonstrsep{s,z},
		vertices in $Z$ are not reachable from $s$ in $\TG - S$, and
		vertices in $A_t$ are not reachable from $s$ in $\TG - S$ before time $t$.
		}
	\label{fig:dp-idea}
\end{figure}
\noindent 
\cref{thm:fpt-tw-tau} is proved by constructing a dynamic program which is based on the fact that for each
vertex~$v \in V\setminus\{s\}$ in a temporal graph $\TG=(V,\TE,\tau)$ there is a
point of time~$t \in [\tau]$ such that~$v$ cannot be reached from~$s \in V$ before time~$t$.
In particular, we guess a 
partition~$V = A_1 \uplus A_2 \uplus \ldots \uplus A_\tau \uplus S \uplus Z$ 
such that (i)~$S$ is the \nonstrsep{s,z}, (ii) in~$\TG-S$~no vertex contained in~$Z$ is reachable from~$s$, and (iii) no vertex~$v \in A_t$ can be reached from~$s$ before time step~$t$, where~$t \in [\tau]$.
See \cref{fig:dp-idea} for an illustrative example.
Due to its length, the formal proof of \cref{thm:fpt-tw-tau} is deferred to \ref{proof:thm:fpt-tw-tau}.

Note that this result implies that \nonstrproblem{} is fixed-parameter tractable when parameterized by $\tw(\ug{\TG}) + \tau$.
  \gappto{\appendixproof}{
    \section{Proof of \cref{thm:fpt-tw-tau}} 
    \label{proof:thm:fpt-tw-tau}
	\renewcommand{\thesection}{4}
	\setcounter{theorem}{6}
\begin{theorem}
	For a given tree decomposition of the underlying graph,
	one can solve \nonstrproblem{} in $\ON((\tau+2)^{\tw(\ug{\TG})+2} \cdot \tw(\ug{\TG}) \cdot |V| \cdot |\TE|)$ time,
	where~$\tau$ is the maximum time label.
\end{theorem}
	\renewcommand{\thesection}{A}
	\setcounter{theorem}{0}
  We prove \cref{thm:fpt-tw-tau} by introducing a dynamic program which is executed on a nice tree decomposition.
\begin{definition}
	  \label{def:nice-tree-dec}
	  A tree decomposition~$\T := (T,(B_i)_{i \in V(T)})$ of a graph~$G$ is a \emph{nice tree decomposition} if $T$ is rooted, every node of the tree~$T$ has at most two children nodes, and for each node~$i\in V(T)$ the following conditions are satisfied:
	  \begin{compactenum}[(i)]
		  \item If~$i$ has two children nodes~$k,j \in V(T)$ in~$T$, then~$B_i = B_k = B_j$. Node~$i$ is called a \emph{join node}.
		  \item If~$i$ has one child node~$j$, then one of the following conditions must hold:
			  \begin{compactenum}
				  \item~$B_i = B_j \cup\{ v \}$. 
				  Node~$i$ is called an \emph{introduce node} of~$v$.
				  \item~$B_i = B_j \setminus \{ v \}$. 
				  Node~$i$ is called a \emph{forget node} of~$v$.
			  \end{compactenum}
		  \item If~$i$ is a leaf in~$T$, then~$|B_i| = 1$. Node~$i$ is called a \emph{leaf node}.
	  \end{compactenum}
	  For the node~$i \in V(T)$, the tree~$T_i$ denotes the subtree of~$T$ rooted at~$i$. %
	  The set~$B(T_i) := \bigcup_{j \in V(T_i)} B_j$ is the union of all bags of~$T_i$.
  \end{definition}

Note that a tree decomposition of width~$\ON(\tw(G))$ for a given graph $G$ with $n$ vertices can be computed in~$2^{\ON{(\tw(G))}} \cdot n$ time~\cite{DBLP:journals/siamcomp/BodlaenderDDFLP16} and can be turned into a nice tree decomposition in polynomial-time \cite[Lemma 7.4]{cygan2015parameterized}.

  We are going to color~$V$ with~$\tau+2$ colors~$\langle A_{[1:\tau]},S,Z \rangle$.
  If a vertex~$v \in V$ has color~$Y \in \{A_{[1:\tau]},S,Z\}$, then we denote this by~$v \in Y$. Thus, formally each color forms a subset of the vertices.
  The meaning of colors is that 
  if~$v \in S$, then~$v$ is in the \nonstrsep{s,z};
  if~$v \in Z$, then~$v$ is not reachable from~$s$ in~$\TG - S$; and
  if~$v \in A_i$, then~$v$ cannot be reached before time point~$i$ from~$s$.

  \begin{definition}
    We say that~$\langle A_{[1:\tau]},S,Z \rangle$ is a \emph{coloring} of~$X \subseteq V(\TG)$ if~$X = A_1 \uplus A_2 \uplus \dots \uplus A_\tau \uplus S \uplus Z$.
    A coloring~$\langle A_{[1:\tau]},S,Z \rangle$ of~$X \subseteq V(\TG)$ is \emph{valid} if
    \begin{inparaenum}[(i)]
	    \item $s \in A_1$,
	    \item $z \in Z$, and
	    \item for all~$a \in A_i$,~$a' \in A_j$, and~$b \in Z$
    \end{inparaenum}
		    \begin{compactitem}
		    \item there is no \nonstrpath{a,b} with departure time at least~$i$ in~$\TG[X] - S$, and
		    \item there is no \nonstrpath{a,a'} with departure time at least~$i$ and arrival time at most~$j-1$ in~$\TG[X] - S$.
		    \end{compactitem}
    
For $Y \supseteq X$, a valid coloring ~$\langle A_{[1:\tau]}',S',Z' \rangle$ of~$Y$ is called an \emph{\exsion} of~$\langle A_{[1:\tau]},S,Z \rangle$ if~$S \subseteq S'$,~$Z \subseteq Z'$, and~$A_i \subseteq A_i'$ for all~$i \in [\tau]$.
If such an \exsion{} exists, $\langle A_{[1:\tau]},S,Z \rangle$ is said to be \emph{\exble} to $Y$.
  \end{definition}

  \begin{lemma}\label{dp-tw-tau:correctness}
	  Let~$\TG = (V,\TE,\tau)$ be a temporal graph, and~$s,z \in V$.
	  There is a valid coloring~$\langle A_{[1:\tau]},S,Z \rangle$ of~$V$ if and only if $S$ is a \nonstrsep{s,z} in~$\TG$.
  \end{lemma}
  \begin{proof}
  \implone{}
  Let~$\langle A_{[1:\tau]},S,Z \rangle$ be a valid coloring of~$V$ such that~$|S| = k$.
  Vertex~$s$ has color~$A_1$ and vertex~$z$ has color~$Z$.
  We know that there is no \nonstrpath{s,z} in~$\TG[V] - S = \TG - S$, otherwise condition~(iii) of the definition of a valid coloring is violated.
  Hence, $S$ is a \nonstrsep{s,z} of size~$k$ in~$\TG$.

  \impltwo{}
  Let~$S$ be a given \nonstrsep{s,z} of size~$k$ in~$\TG$.
  Let~$A \subseteq V(\TG)$ contain all vertices in~$G - S$ that are reachable from~$s$.
  We construct a valid coloring as follows.
  Assign color~$Z$ to all vertices in~$V(\TG) \setminus (A \cup S)$.
  Note that~$z \in Z$.
  For each~$v \in A$ we set~$v \in A_t$
  where~$t \in [\tau]$ is the earliest point of time at which~$v$ can be reached from~$s$.
  In particular~$s \in A_1$.
  As a consequence, there is no~$w \in A_{t'}$ such that there is a \nonstrpath{w,v} with departure time at least~$t'$ and arrival time at most~$t-1$, as otherwise there is a \nonstrpath{s,v} with arrival time at most~$t-1$ contradicting that~$t$ is the earliest time point in which~$v$ is reachable from~$s$.
  Finally, we can observe that there are no~$a \in A_i$ and~$b \in Z$ such that there is a \nonstrpath{a,b} with departure time at least~$i$, because~$a$ can be reached at time point~$i$ from~$s$ and all vertices of color~$Z$ are not reachable in~$\TG - S$.
  Hence,~$\langle A_{[1:\tau]},S,Z \rangle$ is a valid coloring of~$V$.
  \end{proof}

  Let~$\TG = (V,\TE,\tau)$ be a temporal graph,~$s,z \in V$, and~$\T = (T,(B_i)_{i \in V(T)})$ be a nice tree decomposition of~$\ug{\TG}$ of width~$\tw(\ug{\TG})$.
  We add~$s$ and~$z$ to every bag of~$\T$.
  Thus,~$\T$ is of width at most~$\tw(\ug{\TG}) + 2$. 

  In the following, we give a dynamic program on~$\T$.
  For each node~$x$ in~$T$ we compute a table~$D_x$ which stores for each coloring~$\langle A_{[1:\tau]},S,Z \rangle$ of~$B_x$ the minimum size of~$S'$ over all \exsion{}s~$\langle A_{[1:\tau]}',S',Z' \rangle$ of~$\langle A_{[1:\tau]},S,Z \rangle$ to~$B(T_x)$:

  \begin{equation}
	  \label{eq:dp}
D_x[A_{[1:\tau]},S,Z] := 
\min \set*{
\infty, \abs{S'} ;
\parbox[c]{.4\textwidth}{ $\langle A'_{[1:\tau]},S',Z' \rangle$ is an \exsion{} of~$\langle A_{[1:\tau]},S,Z \rangle$ to~$B(T_x)$}}
  \end{equation}
  Let~$r \in V(T)$ be the root of~$T$.
  If~$D_r[A_{[1:\tau]},S,Z] = k' < \infty$, then the coloring~$\langle A_{[1:\tau]},S,Z \rangle$ of~$B_r$ is \exble{} to~$B(T_r)=V(\TG)$ and there is a \nonstrsep{s,z} of size~$k'$ in~$\TG$.
  Hence, the input instance~$(\TG,s,z,k)$ is a \yes-instance if and only if~$k' \leq k$.

  The dynamic program first computes the tables for all leaf nodes of~$T$ and then, in a ``bottom-up'' manner, all tables of nodes of which all child nodes are already computed.
  The computation of~$D_x$,~$x \in V(T)$, depends on the type of~$x$, that is, whether~$x$ is a leaf, introduce, forget, or join node.

  \paragraph{Leaf node}
  Let~$x \in V(T)$ be a leaf node of~$\T$.
  Thus,~$B_x = \{ s, v, z \}$.
  We test each coloring of~$B_x$ and set~$D_x[A_{[1:\tau]},S,Z] = \infty$ if~$s \not \in A_1$ or~$z \not \in Z$, because the coloring cannot be valid.
  Assume~$s \in A_1$ or~$z \in Z$.
  We distinguish three cases.
  \begin{compactenum}[\bf {Case} 1:]
	  \item If~$v \in S$, then this is a valid coloring.
		  We set~$D_x[A_{[1:\tau]},S,Z] := 1$. 
	  \item If~$v \in Z$, then we set
		  $D_x[A_{[1:\tau]},S,Z] := \infty$ if there is a $(\{s,v\},t) \in E(G[B_x])$, and $D_x[A_{[1:\tau]},S,Z]:=0$ otherwise.
	  \item If~$v \in A_i$, $i \in [\tau]$, then we set $D_x[A_{[1:\tau]},S,Z] := \infty$ if there is a $(\{s,v\},t) \in E(G[B_x])$ with $t < i$ or if there is a $(\{v,z\},t) \in E(G[B_x])$ with $i \leq t$, and $D_x[A_{[1:\tau]},S,Z] := 0$ otherwise.
  \end{compactenum}

  \begin{lemma}
	  \label{dp-tw-tau:leaf}
	  Let~$\TG$ be a temporal graph and~$\T$ be a tree decomposition of~$\TG$ as described above,~$x \in V(T)$ be a leaf node, and~$\langle A_{[1:\tau]},S,Z \rangle$ be a coloring of~$B_x$.
	  Then the following holds:
	  \begin{compactenum}[(i)]
		  \item~$D_x[A_{[1:\tau]},S,Z] < \infty$ if and only if~$\langle A_{[1:\tau]},S,Z \rangle$ is a valid coloring of~$B_x$.
		  \item If~$\langle A_{[1:\tau]},S,Z \rangle$ is a valid coloring of~$B_x$, then~$D_x[A_{[1:\tau]},S,Z] = |S|$.
		  \item The table entry~$D_x[A_{[1:\tau]},S,Z]$ can be computed in~$\ON(|\TE|)$ time . 
	  \end{compactenum}
  \end{lemma}
  \begin{proof}
	  We first prove~(i).

	  \impltwo{}
	  Let~$D_x[A_{[1:\tau]},S,Z] = \infty$
	  There are five cases in which~$D_x[A_{[1:\tau]},S,Z]$ is set to~$\infty$.
	  Either~$s \not \in A_1$,~$z \not \in Z$,~$v \in Z$ and there is a time-edge~$(\{s,v\},t) \in E(\TG[B(T_x)])$, or~$v \in A_i$ and there is a time-edge~$(\{s,v\},t) \in E(\TG[B(T_x)])$ with~$t < i$ or there is a time-edge~$(\{v,z\},t) \in E(\TG[B(T_x)])$ with~$i \leq t$, where~$i \in [\tau]$.
	  It follows that~$\langle A_{[1:\tau]},S,Z \rangle$ is no valid coloring of~$B_x$.

	  \implone{}
	  Let~$D_x[A_{[1:\tau]},S,Z] < \infty$.
	  Note that~$s$ must be of color~$A_1$ and~$z$ must be of color~$Z$.
	  Observe that~$D_x[A_{[1:\tau]},S,Z] = 0$ or~$D_x[A_{[1:\tau]},S,Z] = 1$.
	  Consider the case of~$D_x[A_{[1:\tau]},S,Z] = 1$.
	  Thus,~$v \in S$. 
	  This implies that~$\TG[B(T_x)] - S$ is time-edgeless and therefore~$\langle A_{[1:\tau]},S,Z \rangle$ is a valid coloring of~$B_x$.
	  Next, consider the case of~$D_x[A_{[1:\tau]},S,Z] = 0$.
	  If~$v \in Z$, then there is no time-edge from~$s$ to~$v$ which means~$\langle A_{[1:\tau]},S,Z \rangle$ is a valid coloring of~$B_x$.
	  If~$v \in A_i$, then there is no time-edge~$(\{s,v\},t)$ with~$t < i$ and there is no time-edge from~$(\{z,v\},t)$ with~$i \leq t$.
	  In both cases~$\langle A_{[1:\tau]},S,Z \rangle$ is a valid coloring of~$B_x$.
	  
	  If~$\langle A_{[1:\tau]},S,Z \rangle$ is a valid coloring of~$B_x$, then~$D_x[A_{[1:\tau]},S,Z] = |S|$ as we set~$D_x[A_{[1:\tau]},S,Z]=1$ if and only if~$v \in S$.
	  This proves~(ii).
	  Furthermore, we can check by iterating over all time-edges whether~$\langle A_{[1:\tau]},S,Z \rangle$ is a valid coloring of~$B_x$
	  This proves~(iii), and hence (i)--(iii) hold true.
  \end{proof}

  \paragraph{Introduce node} 
  Let~$x \in V(T)$ be an introduce node of~$\T$,~$y \in V(T)$ denote its child node, and~$B_x \setminus B_{y} = \{ v \}$.
  We distinguish three cases.
  \begin{compactenum}[\bf{Case }1:]
	  \item If~$v \in S$, then we set~$D_x[A_{[1:\tau]},S,Z] := D_{y}[A_{[1:\tau]},S \setminus \{v \},Z] +1$.
	  \item If~$v \in Z$, then we set
	    $D_x[A_{[1:\tau]},S,Z]:=D_{y}[A_{[1:\tau]},S,Z\setminus \{v\}]$ if for all~$w \in V$ with~$(\{w,v\},t) \in E(G[B(T_x)])$ it holds that~$w \in A_i \impl t < i$.
	    Otherwise, we set~$D_x[A_{[1:\tau]},S,Z]:=\infty$.
	  \item If~$v \in A_i$, $i \in [\tau]$, then we set
		  $D_x[A_{[1:\tau]},S,Z] :=D_{y}[A_{[1:i-1]},A_i\setminus\{v\},A_{[i+1:\tau]},S,Z]$, if for all~$(\{v,w\},t) \in E(G[B(T_x)])$ it holds that $t \geq i \impl w \in \bigcup_{j=1}^t A_j \cup S$ and $t < i \impl w \in \bigcup_{j=t+1}^{\tau} A_j \cup S \cup Z$.
		  Otherwise, we set $D_x[A_{[1:\tau]},S,Z] := \infty$.

  \end{compactenum}
  We prove the correctness for each case separately.
  We start with the first case.
  \begin{lemma}
	  \label{dp-tw-tau:intro1}
	  Let~$\TG$ and~$\T$ be as described above,~$x \in V(T)$ be an introduce node of~$v$,~$y \in V(T)$ be the child node of~$x$,~$\langle A_{[1:\tau]},S,Z \rangle$ be a coloring of~$B_x$ and~$v \in S$. Then the following holds:

	  \begin{compactenum}
		  \item Coloring~$\langle A_{[1:\tau]},S \setminus \{v \},Z \rangle$ of~$B_{y}$ is \exble{} to~$B(T_{y})$ if and only if coloring~$\langle A_{[1:\tau]},S,Z \rangle$ of~$B_x$ is \exble{} to~$B(T_x)$.
		  \item The value of~$D_x[A_{[1:\tau]},S,Z]$ agrees with \cref{eq:dp} and can be computed in~$\ON(1)$ time.
	  \end{compactenum}
  \end{lemma}
  \begin{proof}
	  \implone{}
	  Let~$\langle A_{[1:\tau]},S \setminus \{v \},Z \rangle$ be a coloring of~$B_{y}$ and~$\langle A'_{[1:\tau]},S',Z' \rangle$ be an \exsion{} to~$B(T_{y})$, where~$\abs{S'} = D_{y}[A_{[1:\tau]},S \setminus \{v \},Z]$.
	  Note that~$v \not \in S'$, because~$v \not \in B(T_{y})$ since~$x$ is the introduce node for~$v$.
	  Since~$B(T_x) \setminus B(T_{y}) = \{v\}$, we know that~$\TG[B(T_{y})] - S'$ is the same temporal graph as~$\TG[B(T_x)] - (S' \cup \{ v \})$.
	  Hence, the coloring~$\langle A_{[1:\tau]},S,Z \rangle$ is \exble{} to~$B(T_x)$ and~$|S' \cup \{v\}| = |S'| + 1$ implies that the table entry~$D_x[A_{[1:\tau]},S,Z] = D_{y}[A_{[1:\tau]},S \setminus \{v \},Z] + 1$.

	  \impltwo{}
	  Let~$\langle A_{[1:\tau]},S \setminus \{v \},Z \rangle$ be not \exble{} to~$B(T_{y})$, then~$\langle A_{[1:\tau]},S,Z \rangle$ is not \exble{} to~$B(T_x)$ because~$\TG[B(T_{y})]$ is a temporal subgraph of $\TG[B(T_x)]$, where~$v \not \in B(T_{y})$.
	  Hence,~$D_x[A_{[1:\tau]},S,Z] = D_{y}[A_{[1:\tau]},S \setminus \{v \},Z] +1 = \infty + 1 = \infty$.

	  Note that~$D_x[A_{[1:\tau]},S,Z]$ can be computed in~$\ON(1)$ time because we just have to look up the value of~$D_{y}[A_{[1:\tau]},S,Z]$. 
  \end{proof}

  Next, we move to the correctness of the second case.
  \begin{lemma}\label{dp-tw-tau:intro2}
	  Let~$\TG$ and~$\T$ be as described above,~$x \in V(T)$ be an introduce node of~$v$,~$y \in V(T)$ be the child node of~$x$,~$\langle A_{[1:\tau]},S,Z \rangle$ be a coloring of~$B_x$ and~$v \in Z$.
	  Then the following holds:
	  \begin{compactenum}
		  \item The coloring~$\langle A_{[1:\tau]},S,Z \rangle$ is \exble{} to~$B(T_x)$ if and only if the coloring~$\langle A_{[1:\tau]},S,$ $Z \setminus \{v \} \rangle$ of~$B_{y}$ is \exble{} to~$B(T_{y})$ 
		        and for all $(\{w,v\},t) \in E(\TG[B(T_x)])$ it holds that~$w \in A_i\impl t < i$.
		  \item The value of~$D_x[A_{[1:\tau]},S,Z]$ agrees with \cref{eq:dp} and can be computed in~$\ON(|\TE|)$ time.
	  \end{compactenum}
  \end{lemma}
  \begin{proof}
	  \implone{}
	  Let the coloring~$\langle A'_{[1:\tau]},S',Z' \rangle$  be an \exsion{} of the coloring~$\langle A_{[1:\tau]},S,Z \rangle$ to~$B(T_x)$.
	  Since~$B(T_{y}) = B(T_x) \setminus \{v\}$ and~$(Z \setminus \{v\}) \subseteq Z \subseteq Z'$, the coloring~$\langle A_{[1:\tau]},S,Z \setminus \{v\} \rangle$ of~$B_{y}$ is \exble{} to~$B(T_{y})$.
	  Furthermore,~$v \in Z$ implies that for all time-edges~$(\{w,v\},t) \in E(\TG[B(T_x)])$ it holds that~$w \in A_i\impl t < i$.

	  \impltwo{}
	  First, if coloring~$\langle A_{[1:\tau]},S,Z\setminus \{v\} \rangle$ of~$B_{y}$ is not \exble{} to~$B(T_{y})$ then coloring~$\langle A_{[1:\tau]},S,Z \rangle$ of~$B_x$ cannot be \exble{} to~$B(T_x)$ because~$\TG[B(T_{y})]$ is a temporal subgraph of~$\TG[B(T_x)]$. 
	  Hence,~$D_x[A_{[1:\tau]},S,Z] = \infty$.

	  Let~$\langle A_{[1:\tau]},S,Z\setminus \{v\} \rangle$ be a coloring of~$B_{y}$ which is \exble{} to~$B(T_{y})$ and for all~$(\{w,v\},t) \in E(\TG[B(T_x)]$ it holds that~$w \in A_i\impl t < i$.
      Then we know that there is an \exsion{}~$\langle A'_{[1:\tau]},S',Z' \rangle$ of~$\langle A_{[1:\tau]},S,Z\setminus \{v\} \rangle$ to~$B(T_{y})$.
	  We claim that~$\langle A'_{[1:\tau]},S',Z'\cup\{v\} \rangle$ is a valid coloring of~$B(T_x)$.
	  Since $\langle A'_{[1:\tau]},S',Z' \rangle$ is a valid coloring of~$B(T_{y})$, we have that~$s \in A_1'$,~$z \in Z'$, and for all~$i,j \in [\tau]$, $a \in A_i'$ and~$a' \in A_j'$ there is no \nonstrpath{a,a'} with departure time at least~$i$ and arrival time at most~$j-1$ in~$\TG[B(T_{y})] - S$.
	  Suppose there exist~$a \in A_i'$ and~$b \in Z'$ such that there is a \nonstrpath{a,b}~$P$ in~$\TG[B(T_x)] - S$ with departure time at least~$i$, for some~$i \in [\tau]$.
	  Since~$B(T_x) \setminus B(T_{y}) = \{v \}$ and $\langle A'_{[1:\tau]},S',Z' \rangle$ is a valid coloring of~$B(T_{y})$, vertex~$v$ is the first vertex of color~$Z$ which is visited by~$P$.
	  Hence, there is a time-edge~$(\{w,v\},t) \in E(\TG[B(T_x)])$ such that~$w \in A_i$ and~$i \leq t$, contradicting~$w \in A_i\impl t < i$.
	  It follows that~$\langle A'_{[1:\tau]},S',Z'\cup\{v\} \rangle$ is a valid coloring of~$B(T_x)$ and hence~$\langle A_{[1:\tau]},S,Z \rangle$ is \exble{} to~$B(T_x)$.
	  Since~$v \in Z$, we have~$D_x[A_{[1:\tau]},S,Z] = D_{y}[A_{[1:\tau]},S,Z\setminus \{v\}]$.

	  Note that~$D_x[A_{[1:\tau]},S,Z]$ can be computed in~$\ON(|\TE|)$ time, since we can decide whether for all~$(\{w,v\},t) \in E(\TG[B(T_x)]$ it holds that~$w \in A_i\impl t < i$ by iterating once over the time-edges in~$\TE$.
  \end{proof}

  Last, we show the correctness of the third case.

  \begin{lemma}\label{dp-tw-tau:intro3}
	  Let~$\TG$ and~$\T$ be as described above,~$x \in V(T)$ be an introduce node of~$v$,~$y \in V(T)$ be the child node of~$x$,~$\langle A_{[1:\tau]},S,Z \rangle$ be a coloring of~$B_x$ and~$v \in A_i$, where~$i \in [\tau]$.
	  Then the following holds:
	  \begin{compactenum}
	    \item	Coloring~$\langle A_{[1:\tau]},S,Z \rangle$ of~$B_x$ is \exble{} to~$B(T_x)$ if and only if
		    coloring~$\langle A_{[1:i-1]},$ $A_i\setminus\{v\},A_{[i+1:\tau]},S,Z \rangle$ of~$B_{y}$ is \exble{} to~$B(T_{y})$ 
		    and for each~$(\{v,w\},t) \in E(\TG[B(T_x)])$ it holds that: $t \geq i\impl w \in \bigcup_{j=1}^t A_j \cup S$ and $t < i\impl w \in \bigcup_{j=t+1}^\tau A_j \cup S \cup Z$.
	    \item The value of~$D_x[A_{[1:\tau]},S,Z]$ agrees with \cref{eq:dp} and can be computed in~$\ON(|\TE|)$ time.
	  \end{compactenum}
  \end{lemma}
  \begin{proof}
	  \implone{}
	  Let~$\langle A_{[1:\tau]},S,Z \rangle$ be a valid coloring of~$B_x$ 
	  and $\langle A'_{[1:\tau]},S',Z' \rangle$ be an \exsion{} to~$B(T_x)$.
	  Since~$B(T_{y}) = B(T_x) \setminus \{v\}$ and~$(A_i \setminus \{v\}) \subseteq A_i \subseteq A_i'$, the coloring~$\langle A_1,A_2,\dots,A_i\setminus \{v\},\dots,A_\tau,S,Z \rangle$ of~$B_{y}$ is \exble{} to~$B(T_{y})$.
	  Let~$(\{v,w\},t) \in E(\TG[B(T_x)])$. 
	  We distinguish two cases.
	  
	  First, let~$t \geq i$.
	  Note that~$w \in B_y$ since~$x$ is an introduce node for~$v$.
	  Since~$\langle A'_{[1:\tau]},S',Z' \rangle$ is a valid coloring of~$B(T_x)$,~$w \not \in Z$ since there is no \nonstrpath{v,w} with departure time~$t$ in~$\TG[B(T_x)] - S'$. 
	  Assume towards a contradiction that~$w \in A_j$, where~$j \in [t+1:\tau]$.
	  Then the time-edge~$(\{v,w\},t)$ is a \nonstrpath{v,w} with departure time at least~$i$ and arrival time at most~$j-1$, contradicting the fact that~$\langle A'_{[1:\tau]},S',Z' \rangle$ is a valid coloring of~$B(T_x)$.
	  Hence,~$w \in \bigcup_{j=1}^t A_j \cup S$.

	  Second, let~$t < i$.
	  Again,~$\langle A'_{[1:\tau]},S',Z' \rangle$ is a valid coloring of~$B(T_x)$ and therefore~$w \not \in \bigcup_{j=1}^{t} A_j$ because otherwise there would be a \nonstrpath{w,v} in~$\TG[B(T_x)] - S'$ with departure time at least~$t$ and arrival time~$t < i$, contradicting the fact that~$\langle A'_{[1:\tau]},S',Z' \rangle$ is a valid coloring.
	  Hence~$w \in \bigcup_{j=t+1}^{\tau} A_j \cup S \cup Z$.

	  \impltwo{}
	  First, if coloring~$\langle A_1,A_2,\dots,A_i \setminus \{ v \},\dots,A_\tau,S,Z \rangle$ of~$B_{y}$ is not \exble{} to~$B(T_{y})$ then coloring~$\langle A_{[1:\tau]},S,Z \rangle$ of~$B_x$ cannot be \exble{} to~$B(T_x)$ because~$\TG[B(T_{y})]$ is a temporal subgraph of~$\TG[B(T_x)]$. 
	  Hence, $D_x[A_{[1:\tau]},S,Z] = \infty$.

	  Let coloring~$\langle A_1,A_2,\dots,A_i \setminus \{ v \},\dots,A_\tau,S,Z \rangle$ of~$B_{y}$ be \exble{} to $B(T_{y})$ 
	  and for each~$(\{v,w\},t) \in E(\TG[B(T_x)])$ it holds that: 
		  $t \geq i\impl w \in \bigcup_{j=1}^t A_j \cup S$ and 
		  $t < i\impl w \in \bigcup_{j=t+1}^{\tau} A_j \cup S \cup Z$.
      Then let~$\langle A'_{[1:\tau]},S',Z' \rangle$ be an \exsion{} of~$\langle A_1,A_2,\dots,A_i \setminus \{ v \},\dots,A_\tau,S,Z \rangle$ to~$B(T_{y})$.
	  We claim that~$\langle A_1',A_2',\dots, A_i' \cup \{ v \},\dots,S',Z' \rangle$ is a valid coloring for~$B(T_x)$.
	  We know~$s \in A_1'$ and~$z \in Z'$.

	  Suppose towards a contradiction that~there exist~$a \in A_j'$ and~$a' \in A_\ell'$, $j,\ell \in [\tau]$, such that there is a \nonstrpath{a,a'}~$P$ with departure time at least~$j$ and arrival time at most~$\ell-1$. 
	  Since coloring~$\langle A'_{[1:\tau]},S',Z' \rangle$ of~$B(T_{y})$ is valid, we know that $v$ appears in~$P$.
	  Thus, there are time-edges~$(\{w_1,v\},t_1),(\{v,w_2\},t_2) \in E(\TG[B(T_x)])$ in~$P$ such that $t_1\leq t_2$ and~$w_1$ appears before~$v$ and~$v$ appears before~$w_2$ in~$P$, where~$w_1 \in A_{u_1}'$,~$w_2 \in A_{u_2}'$.
	  Note that~$w_1 \in A_{u_1}$ and~$w_2 \in A_{u_2}$ as~$x$ is an introduce node of~$v$.
	  Refer to \cref{fig:proof-dp-tw-tau} for an illustration.
	  \begin{figure}[t!]
		  \centering
		  \includegraphics[width=0.8\textwidth]{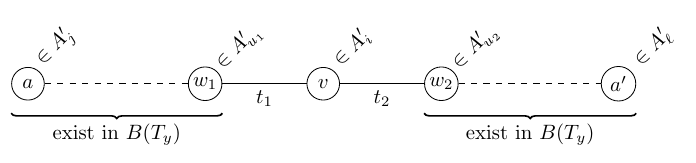}
		  \caption{The \nonstrpath{a,a'}~$P$ from the proof of \cref{dp-tw-tau:intro3}.%
			  }
		  \label{fig:proof-dp-tw-tau}
	  \end{figure}
	  
	  We know the following:
	  \begin{compactitem}
	   \item $u_1 \leq t_1$, otherwise there is a \nonstrpath{a,w_1} with departure time at least~$j$ and arrival time at most~$u_1 - 1$ in~$G[B(T_y)]$, contradicting the fact that~$\langle A'_{[1:\tau]},S',Z' \rangle$ is valid.
	   \item $i\leq t_1$, otherwise either $w_1 \not\in \bigcup_{j=t_1+1}^{\tau} A_j \cup S \cup Z$ contradicting the fact that for each~$(\{v,w\},t) \in E(\TG[B(T_x)])$ it holds that
		  $t < i\impl w \in \bigcup_{j=t+1}^{\tau} A_j \cup S \cup Z$, or $w_1 \in \bigcup_{j=t_1+1}^{\tau} A_j \cup S \cup Z$ and~$w\in A_{u_1}$, contradicting the fact that $\langle A'_{[1:\tau]},S',Z' \rangle$ is a coloring of~$B(T_{y})$.
	   \item $i \leq t_2$, otherwise $i>t_1$ since~$t_1\leq t_2$.
	   \item $u_2 \leq t_2$, otherwise~$i \leq t_2$ and $w_2 \not\in \bigcup_{j'=1}^{t_2} A_{j'}$, contradicting the fact that for each~$(\{v,w\},t) \in E(\TG[B(T_x)])$ it holds that
		  $t < i\impl w \in \bigcup_{j=t+1}^{\tau} A_j \cup S \cup Z$, or $w_2 \in \bigcup_{j'=1}^{t_2} A_{j'}$ and~$w_2 \in A_{u_2}$, contradicting the fact that $\langle A'_{[1:\tau]},S',Z' \rangle$ is a coloring of~$B(T_{y})$.
	  \end{compactitem}
	  It follows that~$P$ contains the \nonstrpath{w_2,a'} as temporal subpath with departure time at least~$u_2 \leq t_2$ and arrival time~$\ell-1$.
	  As this temporal subpath also exists in~$B(T_y)$, this contradicts the fact that coloring~$\langle A'_{[1:\tau]},S',Z' \rangle$ of~$B(T_{y})$ is valid.
	  We conclude that~$P$ does not exist.

	  Next, suppose towards a contradiction that there exist~$a\in A_j'$, $j\in[\tau]$, and~$b \in Z$ such that there is a \nonstrpath{a,b}~$P'$ with departure time at least~$j$. 
	  The vertex~$v \in A_i$ is the last vertex visited by~$P'$ which is not colored by~$Z$, otherwise we would be able to find a subsequence of~$P'$ similar to~$P$.
	  Thus, there are time-edges~$(\{w_1,v\},t_1),(\{v,b\},t_2) \in E(\TG[B(T_x)])$ which are in~$P'$ such that~$w_1$ is visited before~$v$ and~$v$ is visited before~$b$, where~$w_1 \in A_{u_1}'$.
	  We conclude analogously to the case of~$P$ that~$u_1 \leq t_1$,~$i \leq t_1$,~$i \leq t_2$.
	  Since $i \leq t_2$, we have that either~$b \not\in \bigcup_{j=1}^{t} A_j \cup S$, contradicting the fact that for each~$(\{v,w\},t) \in E(\TG[B(T_x)])$ it holds that 
	  $t \geq i\impl w \in \bigcup_{j=1}^t A_j \cup S$, or $b \in \bigcup_{j=1}^{t} A_j \cup S$ and~$b\in Z$, contradicting the fact that~$\langle A_1,A_2,\dots,A_i \setminus \{ v \},\dots,A_\tau,S,Z \rangle$ is a coloring of~$B_{y}$.
	  Hence,~$P'$ does not exist. 

	  Clearly,~$D_x[A_{[1:\tau]},S,Z] = D_{y}[A_1,A_2,\dots,A_i\setminus\{v\},\dots,A_\tau,S,Z]$ because $v \not \in S$.
  
	  Note that~$D_x[A_{[1:\tau]},S,Z]$ can be computed in~$\ON(|\TE|)$ time because we can iterate once over the time-edge set~$\TE$ and decide if for all~$(\{w,v\},t) \in E(\TG[B(T_x)]$ it holds that if~$t \geq i$ then~$w \in \bigcup_{j=1}^t A_j \cup S$ and if~$t < i$ then~$w \in \bigcup_{j=t+1}^\tau A_j \cup S \cup Z$.
  \end{proof}

  \paragraph{Forget node}
  Let~$x \in V(T)$ be a forget node of~$\T$,~$y \in V(T)$ its child, and~$B_{y} \setminus B_x = \{ v \}$.
  We set 
  \[
	  D_x[A_{[1:\tau]},S,Z] = \min \left\{\begin{array}{@{}l@{}} \min_{i \in [\tau]} D_{y}[A_{[1:i-1]},A_i \cup \{ v \}, A_{[i+1:\tau]},S,Z],\\
	  D_{y}[A_{[1:\tau]},S \cup \{v\},Z],\\ D_{y}[A_{[1:\tau]},S,Z\cup \{v\}] \end{array} \right\}.
  \]
  \begin{lemma}\label{dp-tw-tau:forget}
	  Let~$\TG$ and~$\T$ be as described above,~$x \in V(T)$ be a forget node of~$v$,~$y \in V(T)$ be the child node of~$x$, and~$\langle A_{[1:\tau]},S,Z \rangle$ be a coloring of~$B_x$. 
	  Then the following holds:
	  \begin{compactenum}
	  \item The coloring~$\langle A_{[1:\tau]},S,Z \rangle$ of~$B_x$ is \exble{} to~$B(T_x)$ if and only if it has an \exsion{} to $B_{y}$ which is itself \exble{} to~$B(T_{y})$.
		  \item The value of~$D_x[A_{[1:\tau]},S,Z]$ agrees with \cref{eq:dp} and can be computed in~$\ON(|\TE|)$ time.
	  \end{compactenum}
  \end{lemma}
  \begin{proof}
	  \implone{}
	  Let~$\langle A_{[1:\tau]}'',S'',Z'' \rangle$ be an \exsion{} of~$\langle A_{[1:\tau]},S,Z \rangle$ to~$B(T_x)$.
	  Since~$y$ is a child of~$x$ and~$B_x \subseteq B_{y}$, we know that~$B(T_x) = B(T_{y})$ and therefore there is a coloring~$\langle A_{[1:\tau]}',S',Z' \rangle$ of~$B_{y}$ which is \exble{} to~$B(T_{y})$, where~$S' \subseteq S''$,~$Z' \subseteq Z''$, and~$A_i' \subseteq A_i''$, for all~$i \in [\tau]$.
	  It follows from~$B_x \subseteq B_{y}$, that~$S \subseteq S'$,~$Z \subseteq Z'$, and~$A_i \subseteq A_i'$, for all~$i  \in [\tau]$.

	  \impltwo{}
	  It is easy to see that coloring~$\langle A_{[1:\tau]},S,Z \rangle$ of~$B_x$ is \exble{} to~$B(T_x)$ if there is a coloring~$\langle A'_{[1:\tau]},S',Z' \rangle$ of~$B_{y}$ which is \exble{} to~$B(T_{y})$ and is itself an extension of~$\langle A_{[1:\tau]},S,Z \rangle$, because~$\TG[B(T_x)]$ is a temporal subgraph of~$\TG[B(T_{y})]$.
	  Since we want to extend the coloring of~$B_x$ such that we have a minimum size~$S$, we select the minimum over all possible extensions of~$\langle A_{[1:\tau]},S,Z \rangle$ to~$B_{y}$.
	  
	  Note that we can compute the table entry~$D_x[A_{[1:\tau]},S,Z]$ in~$\ON(|\TE|)$ time, because we have to look up~$\tau+2$ entries in~$D_{y}$ and~$\tau \leq |\TE|$, see \cite{zschoche2017computational}.
  \end{proof}
  
  \paragraph{Join node}
  Let~$x \in V(T)$ be a join node of~$\T$,~$y_1,y_2 \in V(T)$ be children of~$x$, and hence~$B_x = B_{y_1} = B_{y_2}$.
  We set $D_x[A_{[1:\tau]},S,Z] := D_{y_1}[A_{[1:\tau]},S,Z] + D_{y_1}[A_{[1:\tau]},S,Z] - |S|$.
  \begin{lemma}\label{dp-tw-tau:join}
    Let~$\TG$ be a temporal graph and~$\T$ be a tree decomposition of~$\TG$ as described above,~$x \in V(T)$ be a join node of~$v$,~$y_1,y_2 \in V(T)$ be the child nodes of~$x$, and~$\langle A_{[1:\tau]},S,Z \rangle$ be a coloring of~$B_x$. 
    Then the following holds:
    \begin{compactenum}
	    \item The coloring~$\langle A_{[1:\tau]},S,Z \rangle$ of~$B_x = B_{y_1} = B_{y_2}$ is \exble{} to~$B(T_x)$ if and only if it is \exble{} to~$B(T_{y_1})$ and~$B(T_{y_2})$.
	    \item The value of~$D_x[A_{[1:\tau]},S,Z]$ agrees with \cref{eq:dp} and can be computed in~$\ON(1)$ time.
    \end{compactenum}
  \end{lemma}
  \begin{proof}

    \implone{}
    Let~$\langle A_{[1:\tau]},S,Z \rangle$ be a coloring of~$B_x = B_{y_1} = B_{y_2}$ 
    and let $\langle A_{[1:\tau]}',S',Z' \rangle$ be an extension to $B(T_x)$.
    Since~$B(T_{y_1}),B(T_{y_2}) \subseteq B(T_x)$ and~$B_x = B_{y_1} = B_{y_2}$, we know that~$\langle A_{[1:\tau]},S,Z \rangle$ is \exble{} to~$B(T_{y_1})$ and~$B(T_{y_2})$.

    \impltwo{}
    Let coloring~$\langle A_{[1:\tau]},S,Z \rangle$ of~$B_x$ be \exble{} to~$B(T_{y_1})$ and~$B(T_{y_2})$.
	Take~$\langle A'_{[1:\tau]},S',Z' \rangle$ and ~$\langle A''_{[1:\tau]},S'',Z'' \rangle$ to be extensions to~$B(T_{y_1})$ respectively to~$B(T_{y_2})$.
    We claim that $\langle A_1' \cup A_1'',A_2' \cup A_2'', \dots, A_\tau' \cup A_\tau'', S' \cup S'',Z'\cup Z'' \rangle$ is a valid coloring of~$B(T_x)$.
    Suppose not, that is, $\langle A_1' \cup A_1'',A_2' \cup A_2'', \dots, A_\tau' \cup A_\tau'', S' \cup S'',Z'\cup Z'' \rangle$ is a coloring but not valid, or it forms no coloring.
    
    In the first case, each~$s\not\in A_1' \cup A_1''$ or $z\not\in Z'\cup Z''$ contradicts the fact that~$\langle A'_{[1:\tau]},S',Z' \rangle$ and~$\langle A''_{[1:\tau]},S'',Z'' \rangle$ are valid colorings.
    Next, suppose there are~$a\in A_i'\cup A_i''$,~$i\in \tau$, and~$b\in Z'\cup Z''$ such that there is a \nonstrpath{a,b}~$P$ with departure time at least~$i$ in~$G[B(T_x)]-(S'\cup S'')$.
    Then, either~$P$ exists in~$G[B(T_{y_1})]$ or in~$G[B(T_{y_2})]$, contradicting the fact that~$\langle A'_{[1:\tau]},S',Z' \rangle$ and~$\langle A''_{[1:\tau]},S'',Z'' \rangle$ are valid colorings, or~$P$ contains an edge~$(\{v,w\},t)$ that is neither in~$G[B(T_{y_1})]$ nor in~$G[B(T_{y_2})]$.
    It follows that~$\{v,w\}\not\in B_{y_1}\cap B_{y,2}$ but~$\{v,w\}\subseteq B_{x}$, contradicting the fact that~$\T$ is a nice tree decomposition.
    It is not difficult to see that the case of~$a\in A_i'\cup A_i''$ and $a\in A_j'\cup A_j''$,~$i,j\in \tau$, such that there is a \nonstrpath{a,a'}~$P$ with departure time at least~$i$ at arrival time at most~$j-1$ in~$G[B(T_x)]-(S'\cup S'')$, follows the same argumentation.
    
    In the second case, that is, $\langle A_1' \cup A_1'', A_2' \cup A_2'', \dots, A_\tau' \cup A_\tau'', S' \cup S'',Z'\cup Z'' \rangle$ forms no coloring, there is a vertex~$v \in B(T_{y_2}) \cap B(T_{y_1})$ which has different colors in~$\langle A'_{[1:\tau]},S',Z' \rangle$ and~$\langle A''_{[1:\tau]},S'',Z'' \rangle$.
    If~$v\not\in B_x= B_{y_1} = B_{y_2}$, then $B^{-1}(v)$ is not a connected subtree of~$T$, contradicting the fact that~$\T$ is a nice tree decomposition.
    If~$v \in B_x$, then~$v$ has different colors in~$\langle A_{[1:\tau]},S,Z \rangle$, contradicting the fact that~$\langle A_{[1:\tau]},S,Z \rangle$ is a coloring of~$B_x$.
    Altogether, it follows that~$\langle A_1' \cup A_1'', A_2' \cup A_2'', \dots, A_\tau' \cup A_\tau'', S' \cup S'',Z'\cup Z'' \rangle$ is a valid coloring of~$B(T_x)$.

    Furthermore, this implies that for all vertices~$w \in B(T_x)$ it holds that $w \in S' \cap S''$ implies~$w \in S$.
    Hence,~$|S'| + |S''| - |S| =|S'| + |S''| - |S' \cap S''| = |S' \cup S''|$.

    Note that by a look up one table entry of~$D_{y_1}$ and one in~$D_{y_2}$, we can compute the table entry~$D_x[A_{[1:\tau]},S,Z]$ in~$\ON(1)$ time. 
  \end{proof}
  
  Having \cref{dp-tw-tau:forget,dp-tw-tau:join,dp-tw-tau:intro3,dp-tw-tau:intro2,dp-tw-tau:intro1,dp-tw-tau:leaf,dp-tw-tau:correctness} at hand, we now prove \cref{thm:fpt-tw-tau}.

  \begin{proof}[Proof of \cref{thm:fpt-tw-tau}]
    The algorithm works as follows on input instance~$\I = (\TG = (V,\TE,\tau),s,z,k)$ of \nonstrproblem{}.
	  Let~$\T$ be a nice tree decomposition for the underlying graph~$\ug{\TG}$ of width at most $\tw(\ug{\TG})$.
    \begin{compactenum}
	    \item Add~$s$ and~$z$ to every bag in~$\ON(\tw(\ug{\TG}) \cdot |V|)$ time.
		    Note that~$|V(T)| \in \ON(\tw(\ug{\TG}) \cdot |V|)$ and that each bag is of size at most~$\tw(\ug{\TG}) + 2$.
	    \item Compute the dynamic program of \cref{eq:dp} on~$\T$.
		    This can be done in~$\ON((\tau+2)^{\tw(\ug{\TG})+2} \cdot \tw(\ug{\TG}) \cdot |V| \cdot |\TE|)$ time because there are at most~$(\tau+2)^{\tw(\ug{\TG})+2}$ possible colorings for each bag, there are at most~$\ON(\tw(\ug{\TG}) \cdot |V|)$ many bags, and table entry for one coloring can be computed in~$\ON(|\TE|)$ time, 
		    see \cref{dp-tw-tau:leaf,dp-tw-tau:intro1,dp-tw-tau:intro2,dp-tw-tau:intro3,dp-tw-tau:forget,dp-tw-tau:join}.

	    \item Iterate over the root table~$D_r$. 
		    If there is an entry of size at most~$k$, then output \yes{}, otherwise output \no{}.
		    The correctness of this step follows from \cref{dp-tw-tau:correctness}.
    \end{compactenum}
    Alltogether, the input instance~$\I$ can be decided in~$\ON((\tau+2)^{\tw(\ug{\TG})+2} \cdot \tw(\ug{\TG}) \cdot |V| \cdot |\TE|)$~time.
  \end{proof}
  }

It remains open whether \nonstrproblem{} is fixed-pa\-ram\-e\-ter tractable when parameterized by~$\tw(\ug{\TG})$ or by~$k+\tw(\ug{\TG})$.
\section{Temporal Restrictions}
\label{sec:tem-classes}

In \Cref{sec:layer,sec:ug-classes} we considered restrictions on the layers and the underlying graph.
Importantly, these restrictions do not cover essential temporal aspects of a temporal graph, that is, any reordering of the layers yields a different temporal graph obeying the same restrictions.
In this section, we study temporal graph classes whose definitions do rely on the order of the layers.
Herein, we study \emph{monotone}, \emph{periodic}, \emph{consecutively connected}, and~\emph{steady} temporal graphs.

Note that the properties \emph{monotone}, \emph{periodic}, and \emph{consecutively connected} yield quite specific temporal graph classes~\cite{casteigts2012time}.
Unfortunately, even on these specific temporal graph classes, except for trivial cases, we obtain hardness by straight-forward arguments.
We refer to~\cref{tab:results} for an overview on our results.
\newcommand{\smtab}[1]{\scriptsize#1}
\newcommand{\mrrb}[2]{\multirow{#1}{*}{\rotatebox[origin=c]{90}{#2}}}
\renewcommand{\arraystretch}{1.1}

\paragraph{Monotone temporal graphs}

 Intuitively, a temporal graph is $p$-monotone if it can be decomposed into $p$ time intervals in each of which the layers are ordered by inclusion.
 
 \begin{definition}
  \label{def:pmonotone}
  A temporal graph $\TG=(V,\TE,\tau)$ is \emph{$p$-monotone} if~$p\in \N$ is the smallest number such that there are~$1=i_1< i_2< \ldots< i_{p+1}=\tau$ such that
  for all~$\ell\in[p]$
  \begin{compactitem}
  	\item $E_j\subseteq E_{j+1}$ for all~$i_\ell\leq j<i_{\ell+1}$, or 
  	\item $E_j\supseteq E_{j+1}$ for all~$i_\ell\leq j<i_{\ell+1}$
  \end{compactitem}
  holds.
 \end{definition}

 \citet{KhodaverdianWWY16} call a temporal graph monotone if whenever an edge is contained in a layer, this edge is contained in all succeeding layers. 
 Their motivation is based on temporal graphs that model activation of proteins or, more generally, activation by connected components. 
 Note that their definition of monotone temporal graphs is equivalent to our definition of 1-monotone temporal graphs where each layer is a subgraph of its successor.
 
\begin{table}[t]
  \setlength{\tabcolsep}{8pt}
  \centering
  \caption
  { 
    Summary of the results of \cref{sec:tem-classes}, where~$\tau$ denotes the maximum time label and~$r$ the number of periods in~$\TG$.
  }
  \def\fbs{3.5cm}
  \begin{tabular}{@{}p{0.49\textwidth}p{0.28\textwidth}p{0.125\textwidth}@{}}%
  \toprule
						& \multicolumn{2}{c}{\nonstrproblem} \\\cmidrule{2-3}
						& polynomial-time 					& \NP{}-hard \\\cmidrule{1-3}
  $p$-monotone temporal graphs					& $p=1$						& $p\geq 2$\\
  $p$-periodic	temporal graphs		 		& $p=1$, or $r\geq n$ 				& $p\geq 2$ \\
  $T$-interval connected temporal graphs			& ---							& $T\geq 1$ \\
	  $\lambda$-steady temporal graphs				& $\lambda =0$ or ($\lambda,\tau$ const.)       & $\lambda\geq 1$ \\
    \bottomrule	
  \end{tabular}
  \label{tab:results}
\end{table}

If a temporal graph $\TG$ has a layer $G_i = \ug{\TG}$,
then \nonstrproblem{} can trivially be solved by finding an $(s,z)$-separator in $G_i$.
In that case we call $\TG$ \emph{grounded}.
Therefore, a straightforward application of the folklore Ford-Fulkerson algorithm gives the following: 

 \begin{observation}
  \label{obs:easyonmonotone}
  \nonstrproblem{} is solvable in $\ON(k\cdot |\TE|)$ time on grounded temporal graphs,
  where $k$ is the solution size and $|\TE|$ the number of time-edges.
 \end{observation}

Note that 1-monotone temporal graphs are always grounded.
However, the situation changes already when the temporal graph is~$2$-monotone but not grounded.
To see that, first note that one can make every temporal graph $\tau$-monotone by simply adding edge-free layers between any two consecutive layers, formally:
 
 \begin{observation}
  \label{obs:nphardon2mono}
    There is a polynomial-time many-one reduction that maps any instance $(\TG=(V,\TE,\tau),s,z,k)$ of~\nonstrproblem{} to an equivalent instance~$(\TG':=(V,\TE',2\tau-1),s,z,k)$ such that for all~$i\in[\tau]$ it holds that~$E_{2i-1}(\TG')=E_i(\TG)$ and for all~$i\in[\tau-1]$ it holds that~$E_{2i}(\TG')=\emptyset$.
 \end{observation}
 
 As~\nonstrproblem{} is already~\NP{}-complete for~$\tau=2$ \cite{zschoche2017computational}, this yields the following.
 
 \begin{observation}
  \label{obs:nphardon2mono}
    For all~$p\geq 2$, \nonstrproblem{} on $p$-monotone temporal graphs is \NP{}-complete.
 \end{observation}

\paragraph{Periodic temporal graphs}

 In several real-world scenarios one observes periodicity;
 indeed, whenever one observes mobile entities with periodic movements~\cite{casteigts2012time}, such as satellites or (scheduled) public transport, over longer time periods, periodic patterns appear.
 Such models motivate the following class of temporal graphs.
 
 \begin{definition}
  A temporal graph~$\TG=(V,\TE,\tau)$ is \emph{$p$-periodic}
	 if~$p\in\N$ is the smallest number such that~$\TG = {\TG'}^{r}$ for some~$\TG'=(V,\TE',p)$ and $r$ is called the \emph{number of periods}.
 \end{definition}

 Different notions of periodic temporal graphs exist in the literature.
 \citet{FlocchiniMS13} consider periodic temporal graphs obtained from ``carriers'', that is, a set of strict temporal paths define a network.
  \citet{LiuW09a} consider delay-tolerant networks where vertices have some cyclic movement pattern and get connected when they are in reach: if the time steps are large enough, then periodicity is observed.
 In both cases, the smallest common multiple of the time spans of the entities define the length of a period.
  \citet[Class~8]{casteigts2012time} define periodic temporal graphs by periodicity of edges, that is, for all edges~$e$, time steps~$t$, and~$c\in\N$, edge~$e$ is present at time step~$t$ if and only if~$e$ is present at time step~$t+c\cdot p$, where~$p$ is the periodicity.
  They require the underlying graph to be connected, but they do not require minimality on the periodicity.

 We know that \nonstrproblem{} is \NP{}-complete on~$2$-pe\-ri\-od\-ic temporal graphs~\cite{zschoche2017computational}.
 Contrarily, on $1$-periodic temporal graphs, \nonstrproblem{} collapses to \sepproblem{} in the underlying graph.
 Surprisingly, if the number of periods is large enough, then the problem becomes polynomial-time solvable.

	Let $P$ be an \npath{s,z} of length $\ell$ in the underlying graph $\ug{\TG}$ of the temporal graph $\TG = (V,\TE, \tau)$.
	A sequence $P' = \big((e_1,t_1),(e_2, t_2),\dots,(e_\ell,t_\ell)\big)$ of~$\ell$ time-edges from $\TE$ is a \emph{realization} of~$P$ ($P' \simeq P$) if $\big(e_1,e_2,\dots,e_\ell\big)$ is~$P$.
	Note, that the sequence of labels of $P'$ is not necessarily non-decreasing.
	Intuitively, we want measure how many non-decreasing points a realization of $P$ must have.
	The \emph{distance to temporality} of~$P$ in~$\TG$ is~$\min_{P' \simeq P} |f_{P'}|-1$, where~$|f_{P'}|$ is the number of monotonically increasing intervals of the function~$f_{P'} : [\ell] \rightarrow [\tau], f_{P'}(x) := t_x$ where~$t_x$ is the label of the~$x$-th time-edge of~$P'$. 
	Furthermore, the distance to temporality from~$s$ to~$z$ in~$\TG$ is the maximum distance to temporality over all \npath{s,z}s in~$\ug{\TG}$.
 
 \begin{lemma}
  \label{lem:highperiodicity}
	 Let~$\TG= {\TG'}^r$ be a $p$-periodic temporal graph such that the number of periods $r$ is at least the distance to temporality from $s$ to $z$ in $\TG'$.
  	 Then \nonstrproblem{} is solvable in $\ON(k\cdot |\TE|)$ time,
  where $k$ is the solution size and $|\TE|$ the number of time-edges. 
  \end{lemma}
 \begin{proof}
	 Let~$\TG= {\TG'}^r$ be a $p$-periodic temporal graph such that the number of periods $r$ is at least the distance to temporality from $s$ to $z$ in $\TG'$.
	 Then every~\npath{s,z} in~$\ug{\TG}$ forms a \nonstrpath{s,z} in~$\TG$.
	 Hence, we can compute a minimum~\nsep{s,z} in~$\ug{\TG}$, by $k$ rounds of the Ford-Fulkerson algorithm, to solve~\nonstrproblem{}. 
 \end{proof}
	Observe that the distance to temporality from~$s$ to~$z$ is two in the temporal graph from the reduction of \citet{zschoche2017computational} for maximum label~$\tau=2$.
  	Thus \nonstrproblem{} is \NP-complete, even if the input temporal graph~$\TG= {\TG'}^r$ is~$p$-periodic and the number of periods $r$ is one less than the distance to temporality from~$s$ to~$z$ in~$\TG'$.

However, the distance to temporality is clearly upper-bounded by the number of vertices.
Hence, we obtain the following.

\begin{corollary}
 Let~$\TG=(V,\TE,\tau)$ be a $p$-periodic temporal graph.
 If the number of periods~$r \geq |V|$, then \nonstrproblem{} is solvable in $\ON(k\cdot |\TE|)$ time,
  where $k$ is the solution size and $|\TE|$ the number of time-edges. 
\end{corollary}

\paragraph{Interval-connected temporal graphs}

\citet[Definition~2.1]{KuhnLO10} introduced the following class of temporal graphs.

\begin{definition}
 A temporal graph~$\TG=(V,\TE,\tau)$ is~$T$-interval connected for~$T\geq 1$ if for every~$t\in[\tau-T+1]$ the static graph~$G:=(V,\bigcap_{i=t}^{t+T-1} E_i(\TG))$ is connected.
\end{definition}

\citet{KuhnLO10} studied $T$-interval connected temporal graphs in the context of counting and token dissemination.
Note that temporal graphs where each layer is connected are 1-interval connected temporal graphs, but are not necessarily~$T$-interval connected for some~$T\geq 2$.
On the contrary, for every~$T$-interval connected temporal graph it holds that each layer is connected.

\begin{observation}
 \label{obs:tintconhard}
 There is a polynomial-time many-one reduction that maps any instance $(\TG=(V,\TE,\tau),s,z,k)$ of~\nonstrproblem{} to an equivalent instance~$(\TG'=(V',\TE',\tau),s,z,k+1)$ such that~$\TG'$ is $T$-interval connected for every~$T\geq 1$. 
\end{observation}

\begin{proof}
 Let instance~$\I=(\TG=(V,\TE,\tau),s,z,k)$ of \nonstrproblem{} be given.
 Obtain the temporal graph~$\TG'$ from~$\TG$ by adding a vertex~$v$ to~$\TG$ and connect~$v$ to all other vertices in~$V$ in each layer of~$\TG$.
 Clearly, every temporal~$(s,z)$-separator in~$\TG'$ contains vertex~$v$.
 As~$\TG=\TG'-v$, instance~$(\TG',s,z,k+1)$ is equivalent to~$\I$.
 Moreover, for any~$T\geq 1$ and~$t\in[\tau-T+1]$ the graph~$G:=(V,\bigcap_{i=t}^{t+T-1} E_i(\TG))$ is a supergraph of the star graph with center~$v$ and set~$V$ of leaves.
\end{proof}

\paragraph{Steady temporal graphs}

When observing a network over time with high resolution, we expect evolutionary instead of revolutionary changes in each time step.
For instance, observing any contact network per second, we do not expect many contacts to appear in the same second.
More generally, in several real-world scenarios we do not expect big changes from one time step to the other.
This assumption motivates the following class of temporal graphs.

\begin{definition}
	A temporal graph~$\TG = (V,\TE,\tau)$ is \emph{$\lambda$-steady} if~$\lambda\in\N$ is the smallest number such that for each point in time~$t \in [\tau-1]$ the size of the symmetric difference of two consecutive edge sets~$|E_t \triangle E_{t+1}|$  is at most~$\lambda$.
\end{definition}
To the best of our knowledge, this class has not been considered in the literature.

The following shows that many hardness results for temporal graphs are also valid for $\lambda$-steady temporal graphs, even if~$\lambda = 1$.

\begin{proposition}
	\label{lemma:steady-hardness-translator}
	There is a polynomial-time many-one reduction that maps any instance $(\TG=(V,\TE,\tau),s,z,k)$ of~\nonstrproblem{} to an equivalent instance~$(\TG'=(V',\TE',\tau'),s,z,k)$ such that~$\TG'$ is 1-steady.
\end{proposition}
\begin{proof}
Let~$\I=(\TG=(V,\TE,\tau),s,z,k)$ be an instance of \nonstrproblem{}.
We define~$\TG' = (V,\TE',\tau')$ as follows.
Intuitively, we slowly construct and subsequently deconstruct each layer~$E_i$ of~$\TG$.
Formally, for each~$i\in[\tau]$ we write $E_i := E_i(\TG) =: \set{(e_{i,j}, i) ; j \in [\abs{E_i}]}$ 
and define an auxiliary temporal graph $\TG_i := (V, \TE_i, 2\cdot\abs{E_i})$
where $\TE_i := \set{(e_{i,j}, t); j \in [\abs{E_i}] \land \abs{\abs{E_i}-t} < j}$.
In particular, we have $\abs{E_1(\TG_i)} = 1$,
$E_{2\cdot\abs{E_i}}(\TG_i) = \emptyset$,
and $E_{\abs{E_i}}(\TG_i) = E_i$.
Now we construct $\TG'$ as $\TG' := \TG_1 \circ \TG_2 \circ \ldots \circ \TG_{[\tau]}$.
Observe that~$\TG'$ is $1$-steady.
Furthermore, for any \nonstrpath{s, z} in~$\TG'$, there is a \nonstrpath{s,z} in $\TG$ that uses the same vertices and vice versa.
Hence \nonstrproblem{} is equivalent on inputs $\TG$ and $\TG'$.
\end{proof}
The reduction of \cref{lemma:steady-hardness-translator} increases the maximum label by a factor depending on the input size.
Indeed, from previous results \cite{zschoche2017computational} it follows that for any fixed $\lambda$, \nonstrproblem{} on~$\lambda$-steady temporal graphs is fixed-parameter tractable when parameterized by~$\tau$.
\begin{corollary}
 	\label{thm:fpt-steady-tau}
 	For any fixed $\lambda$ we have that \nonstrproblem{} on~$\lambda$-steady temporal graphs is fixed-parameter tractable when parameterized by the maximum label~$\tau$. 
\end{corollary}
\begin{proof}
For a temporal graph $\TG = (V,\TE,\tau)$, the vertex 
set~$W:=\{v\in V\mid \exists \{v,w\}\in (\bigcup_{i=1}^\tau E_i) \setminus (\bigcap_{i=1}^\tau E_i)\} \subseteq V$
  is called the \emph{temporal core} of $\TG$.
	\citet{zschoche2017computational} showed that \nonstrproblem{} is fixed-parameter tractable when parameterized by the size of the temporal core.

	The statement follows directly from the fact that the temporal core of a~$\lambda$-steady temporal graph~$\TG = (V,\TE,\tau)$ is upper-bounded by~$2 \cdot \lambda \cdot \tau$.
\end{proof}

\section{(Almost) Order-Preserving Temporal Unit Interval Graphs}
\label{sec:unitinterval}
In this section,
we sort of combine aspects studied in \Cref{sec:ug-classes} (restrictions of the underlying graph)
and \Cref{sec:tem-classes} (temporal restrictions).
To this end, we focus on temporal graphs where
each layer is a unit interval graph and we further restrict how much the intervals may change over
time. 
This is a layer-wise restriction with, additionally, a temporal restriction.
Recall from \cref{thm:uinterval-hardness-fixed-tau} that 
\nonstrproblem{} remains \NP-complete on temporal graphs where each layer is a 
unit interval graph, even if the maximum label $\tau$ is a small constant.

Now we show in the
following that if there is an ordering on the vertices that matches the relative
positions of the intervals in all layers, then we can solve \nonstrproblem{} in
polynomial time.
We then generalize this by introducing a parameter that, informally speaking,
describes how much the interval orderings may change over time, and show
fixed-parameter tractability with respect to the combination of this new
parameter and the maximum label~$\tau$.

We call a total ordering~$<_V$ on a vertex set~$V$ \emph{compatible} with a unit interval graph~$G=(V,E)$ if there are unit intervals~$[a_v,a_v+1]$ with~$a_v\in\mathbb{R}$ for all vertices~$v\in V$ that induce the graph $G$ and for all~$u, v\in V$ with~$u<_V v$ we have that~$a_u\le a_v$. Note that for every unit interval graph there is a total ordering on the vertices that is compatible with it.

\begin{definition}
  A temporal graph~$\TG=(V,\TE,\tau)$ is an \emph{\outempinterval} if~$\TG$ is a \utempinterval{} and there is a total ordering~$<_V$ on the vertex set~$V$ that is compatible with every layer~$G_i$.
\end{definition}

Given an \outempinterval{} $\TG=(V,\TE,\tau)$, 
we denote by~$<_V$ a compatible total ordering on~$V$.
Let~$n:=|V|$, 
and number the vertices in~$V =: \{v_1, v_2,\ldots, v_n\}$
such that $v_i<_V~v_j \Leftrightarrow i\le j$.
Furthermore, we use the following notation:~$V_{<i} := \{v_j \mid 1\leq j<i\}$ and~$V_{>i} := \{v_j \mid n\geq j>i\}$
and~$N^{>}_{G_t}(v_i):=N_{G_t}(v_i)\cap V_{>i}$. 
If the ordering~$<_V$ is clear from the context,
then we refer to vertices as smaller or larger than other vertices to express that they appear before or after, respectively, in the ordering~$<_V$.

\begin{lemma}\label{thm:finding-uinterval-order}
Order-preserving \utempinterval{}s can be recognized in polynomial time and a compatible vertex ordering for a given \outempinterval{} can be computed in polynomial time.
\end{lemma}
\begin{proof}
Let $\TG = (V, \TE, \tau)$ be a temporal graph.
Then, due to \citet[Theorem~1]{LOOGES199315}, we know that $\TG$ is an
\outempinterval{} with vertex ordering $<_V$ if and only if the vertices in of
every closed neighborhood $N_{G_i}[v] := N_{G_i}(v) \cup \{ v \}$ with~$v
\in V$ of every
layer $i \in [\tau]$ appear consecutively in the ordering $<_V$.
Thus, the problem can be solved by searching a column ordering of the matrix $M \in \{0,1\}^{\abs{V}\cdot \tau \times \abs{V}}$
defined by $M[(i,t),j] = 1 \iff v_j \in N_{G_t}[v_i]$ that has the \emph{consecutive ones property}, a task for which a linear-time algorithm is known \cite{BOOTH1976335}.
\end{proof}

We now state some useful properties of temporal paths and separators in \outempinterval{}s.
Due to \condRef{item:intervalsep0} of the following lemma, we will henceforth assume without loss of generality that~$v_1=s$ and~$v_n=z$.

\begin{lemma}
 \label{lem:metaordpreinterv}
 Let $\TG=(V,\TE,\tau)$ be an \outempinterval{} with ordering~$<_V$. %
 \begin{compactenum}[(i)]
  \item\label{item:hered} For all $1\le a \le b\le\tau$ and for all $S\subseteq V$ we have that $\TG_{[a:b]}-S$ is also an \outempinterval{}.
  \item\label{item:intervalpaths} If for some $1\le i<j\le n$ there is a \nonstrpath{v_i,v_j} $P$ in $\TG$, then there is \nonstrpath{v_i,v_j} $P'$ in $\TG$ that visits its vertices in the order given by~$<_V$.
  \item\label{item:intervalsep0} Let $S\subseteq V$ be a \nonstrsep{v_i,v_j} in $\TG$ for some $1\le i < j\le n$. 
Then $S':=S\setminus (V_{<i}\cup V_{>j})$ is also a \nonstrsep{v_i,v_j} in $\TG$.
  \item\label{item:intervalsep1} A \nonstrsep{v_i,v_j} in~$\TG$ is also a \nonstrsep{v_{i'},v_{j'}} in~$\TG$ for all $1\le i'\le i< j\le j'\le n$.
  \item\label{item:intervalsep2} Let $S\subseteq V\setminus\{s,z\}$ such that
  $v_i$ is the largest vertex reachable from~$s$ in~$\TG-S$.
	Let $t$ denote the first time~$v_i$ is reachable from~$s$ in $\TG-S$, and let~$t\leq t'\leq \tau$. 
	Then $N_{G_{t'}}^>(v_i)\subseteq S$.
  \item\label{item:intervalsep3} Let $S_1\subseteq V\setminus\{s,z\}$ such that $v_i$ is
  the largest vertex reachable from~$s$ in~$\TG_{[1:t]}-S_1$ for some
  $t\in[\tau-1]$. Let $S_2\subseteq V\setminus\{s,z\}$ such that $v_j$ is
  the largest vertex reachable from~$s$ in $\TG_{[t+1:\tau]}-S_2$. If $i\le j$,
  then $S:=S_1\cup S_2$ is a \nonstrsep{s,z} in~$\TG$ such that there is no
  vertex reachable from $s$ in $\TG-S$ that is larger than $v_j$.
  \item\label{item:intervalsep4} Let $S\subseteq V$ be an inclusion-wise minimal
  \nonstrsep{s,z} in $\TG$ with the property that a given $v_i$ is the largest
  vertex that is reachable from $s$ in~$\TG-S$ and let $v_j$ be the smallest
  vertex that is not in $S$ such that~$S$ is also a \nonstrsep{s,v_j} in $\TG$.
  Then for all $v_i<_V v<_Vv_j$ with~$v_i\neq v\neq v_j$ we have that
  $v\in S$, and we have that $S\cap V_{>j}=\emptyset$.
 \end{compactenum}
\end{lemma}

\begin{proof}
\condRef{item:hered}: Obvious.

\smallskip\noindent\condRef{item:intervalpaths}:
We prove that there is a \nonstrpath{v_i,v_j} $P'$ in $\TG$ that visits its vertices in the order given by~$<_V$ and $t\leq t'$, where~$t$ and~$t'$ denote the first time label in~$P$ and in~$P'$, respectively.
We give an inductive proof over the number of edges in the temporal~$(v_i,v_j)$-path~$P$.
For the base case, if~$P$ has only one edge, then~$E(P)=(\{v_i,v_j\},t)$ for some~$t\in[\tau]$.
Hence,~$P':=P$ clearly is the sought temporal path.
Now, assume that the statement holds for all temporal~$(v_i,v_j)$-paths with at most~$\ell\in \N$ edges.
For the inductive step, let~$P$ be a temporal~$(v_i,v_j)$-path with exactly~$\ell+1$ edges.
Let~$v_{i'}$ be the last vertex on~$P$ such that~$i'\leq i$, and let~$t\in[\tau]$ be the index of the layer where~$P$ contains the edge~$\{v_{i'},v_x\}$, where~$v_x$ is the successor of~$v_{i'}$ on~$P$.
Since~$G_t$ is a unit interval graph with order~$<_V$, the edge~$\{v_i,v_x\}$ is present in~$G_t$.
Denote by~$P_x$ the temporal $(v_x,v_j)$-subpath of~$P$, starting at vertex~$v_x$.
Observe that~$P_x$ has at most~$\ell$ edges, and hence there is a path~$P_x'$  visiting its vertices in the order given by~$<_V$ and starting at some time label~$t'\geq t$.
Thus, the path~$P'=(\{v_i\}\cup V(P_x'),\{(\{v_i,v_x\},t)\}\cup E(P_x'))$, that starts with edge~$(\{v_i,v_x\},t)$ and then follows~$P_x'$, visits its vertices in the order given by~$<_V$ and starts at time label~$t$ being at least the first time label appearing on the edges of~$P$.

\smallskip\noindent\condRef{item:intervalsep0}:
Follows directly from \condRef{item:intervalpaths}.

\smallskip\noindent\condRef{item:intervalsep1}:
Follows directly from \condRef{item:intervalpaths}.

\smallskip\noindent\condRef{item:intervalsep2}:
Suppose not. 
Then there is a time step~$t''$ with larger neighborhood and hence there is a vertex $v_j\in N_{G_{t''}}^>(v_i)\setminus N_{G_{t'}}^>(v_i)$.
Hence,~$v_j$ with $j>i$ is reachable from~$s$ in $\TG_{[1:t'']}-S$, contradicting the definition of~$v_i$.

\smallskip\noindent\condRef{item:intervalsep3}:
Follows directly from \condRef{item:intervalpaths}.

\smallskip\noindent\condRef{item:intervalsep4}: Assume towards a contradiction that there is a vertex $v\notin S$ with $v_i<_V v<_Vv_j$. 
Then either $v$ is reachable from $s$ in $\TG-S$, which would be a contradiction to $v_i$ being the largest vertex reachable from $s$ in $\TG-S$, or~$v$ is not reachable from $s$ in $\TG-S$, a contradiction to the assumption that~$v_j$ is the smallest vertex such that $S$ is also a \nonstrsep{s,v_j} in~$\TG$. 
Furthermore, $S\cap V_{>j}=\emptyset$ follows from the assumption that $S$ is inclusion-wise minimal and \cref{lem:metaordpreinterv}\condRef{item:intervalsep0}.
\end{proof}

Now we have the necessary tools to prove that \nonstrproblem{} can be solved in polynomial time on \outempinterval{}s.

\begin{theorem}
  \label{prop:intervalpoly}
  \nonstrproblem{} on \outempinterval{}s is solvable in $\ON(|V|^2\cdot \tau^2)$ time.
\end{theorem}
\begin{proof}
Let $\TG=(V,\TE,\tau)$ be a given \outempinterval{} and $k$ be a given upper bound on the temporal separator size.
By \cref{thm:finding-uinterval-order} we can find a total vertex ordering $<_V$ compatible with every layer.
Assume that there is no layer with an edge between~$s$ and~$z$. 
In order to solve the problem, we use the
following dynamic programming table~$T$ of size~$\tau\times (n-1)$.
In the table entry $T[t,i]$ we store a minimum \nonstrsep{s,z}~$S$ for
$\TG_{[1:t]}$ with the property that there is no vertex reachable from~$s$ in~$\TG_{[1:t]}-S$ that is larger than~$v_i$.
Let $$\mathcal{N}(v,t,t'):=\begin{cases}
      \set*{N^{>}_{G_{t''}}(v)\mid t\leq t''\le t'},& \text{if } \forall t\leq
      t''\le t': (\{v,z\},t'')\notin\TE, \\
      \set*{V\setminus\{s,z\}},& \text{otherwise.}
    \end{cases}$$
Let~$T$ be defined in the following way:
\begin{align}
T[1,1] :=& \ N_{G_1}(s),\label{DP:base1}\\
	T[t,1] :=& \argmax_{S\in \mathcal{N}(s,1,t)}|S|,\label{DP:base2}\\
T[1,i] :=& \argmin_{S\in Y_i}|S|,\,\text{where
}Y_i:=\{T[1,i-1]\}\cup \mathcal{N}(v_i,1,1),\label{DP:induction1}\\
 T[t,i] :=& \ \argmin_{S\in X_{t,i}} |S| \text{, where}\label{DP:induction2}\\
 & \ X_{t,i} := \set*{T[t',i']\cup \argmax_{{S\in \mathcal{N}(v_i,t'+1,t)}} \abs{S} ; i'\in[i-1] \wedge t'\in[t-1] }\nonumber\\
& \ \ \ \ \ \ \ \ \ \ \ \cup \set*{T[t,i-1] } \cup \set*{\argmax_{S\in
\mathcal{N}(v_i,1,t)} \abs{S} }\nonumber.
\end{align}
We decide whether we face a \yes-instance by checking if there is an $i\in[n-1]$ such that~$|T[\tau,i]|\le k$.

It is easy to see that each table entry can be computed in $\ON(|V|\cdot\tau)$ time and the table has size $|V|\cdot \tau$. 
Hence, the algorithm has the claimed polynomial running time. 

\emph{Correctness.} 
 We prove by induction on both dimensions of~$T$ that~$T[t,i]$ is a minimum
 \nonstrsep{s,z}~$S$ for $\TG_{[1:t]}$ with the property that there is no vertex
 reachable from~$s$ in~$\TG_{[1:t]}-S$ that is larger than~$v_i$ with respect
 to $<_V$.
 First, observe that
 \cref{lem:metaordpreinterv}\condRef{item:intervalsep2} implies that~$T[1,1]$
 and~$T[t,1]$ are correctly filled in~\cref{DP:base1,DP:base2}.
Hence, the base for our induction is correct.
 
We proceed with the proof of the cases specified by~\cref{DP:induction1,DP:induction2} in two steps. 
First we show that for all~$T[t,i]$ with~$t\ge 1$ and~$i>1$, we have that~$T[t,i]$ is a \nonstrsep{s,z}~$S$ for~$\TG_{[1:t]}$ with the property that there is no vertex reachable from~$s$ in~$\TG_{[1:t]}-S$ that is larger than~$v_i$. 
Then, in a second step, we show that said separator is \emph{minimum}.
 
It is easy to check that if $t=1$, then for all $i\in[n-1]$ we have that~$T[1,i]$ (as specified in \cref{DP:induction1}) is a \nonstrsep{s,z} with the desired
properties. Next, we consider the case that $t,i>1$. We show that every set in~$X_{t,i}$ is a \nonstrsep{s,z} with the desired
properties. By induction we know that this holds for~$T[t,i-1]$. It is also easy to
check that it holds for~$S':=\argmax_{S\in \mathcal{N}(v_i,1,t)} |S|$. 
For arbitrary~$i'\in[i-1]$ and $t'\in[t-1]$ (\cref{DP:induction2}) it is also straightforward to see that~$S':=T[t',i']\cup \argmax_{S\in \mathcal{N}(v_i,t'+1,t)} |S|$ has the desired properties. 
By induction, $T[t',i']$ contains a \nonstrsep{s,z} for~$\TG_{[1:t']}$ with the property that there is no vertex reachable from~$s$ in~$\TG_{[1:t']}-T[t',i']$ that is larger than~$v_{i'}$. 
The set~$S'':=\argmax_{S\in \mathcal{N}(v_i,t'+1,t)} |S|$ either equals~$V\setminus\{s,z\}$, in which case we clearly have a separator with the desired properties, or it forms a \nonstrsep{s,z} for~$\TG_{[t'+1:t]}$ with the property that there is no vertex reachable from~$s$ in~$\TG_{[t'+1:t]}-S'$ that is larger than~$v_i$. 
Then by \cref{lem:metaordpreinterv}\condRef{item:intervalsep3} we get that we have a separator with the desired properties. 
 
 Now we show that for all $t\ge 1$ and $i>1$, the separator contained
 in~$T[t,i]$ is of minimum size.
 Let~$S^\star\subseteq V\setminus\{s,z\}$ be a minimum set of vertices such that
 in~$\TG_{[1:t]}-S^\star$ the vertex~$v_j$, $j\leq i$,  is the largest reachable
 vertex from~$s$.
 If~$j<i$, then by induction hypothesis (both for $t=1$ and $t>1$) we have that~$|S^\star|\geq |T[t,i-1]|$ and hence~$|T[t,i]|\le |S^\star|$.
 
 We continue with the case that~$j=i$.
 If $v_i$ is reachable in $\TG_{[1:1]}-S^\star$ from~$s$, then
 by~\cref{lem:metaordpreinterv}\condRef{item:intervalsep2} we know that
 $N^{>}_{G_{t'}}(v_i)\subseteq S^\star$ for all~$t'\in[t]$.
 As~$S^\star$ is minimum, it holds that~$|S^\star|=\max_{S\in
 \mathcal{N}(v_i,1,t)} |S|$, and we have that $\argmax_{S\in \mathcal{N}(v_i,1,t)} |S|\in X_{t,i}$
 (if $t=1$, then $\argmax_{S\in \mathcal{N}(v_i,1,t)} |S|\in
 Y_i$) which implies that $|T[t,i]|\le |S^\star|$.
 
 Now assume that~$t>1$ and~$v_i$ is not reachable from~$s$ in~$\TG_{[1:1]}-S^\star$.
 Let~$t'$ be the largest time-step in which~$v_i$ is not reachable from~$s$ in
 $\TG_{[1:t']}-S^\star$, and let~$i'<i$ be the largest index such
 that~$v_{i'}$ is reachable from~$s$ in~$\TG_{[1:t']}-S^\star$.
 By~\cref{lem:metaordpreinterv}\condRef{item:intervalsep2}, we know
 that~$S'':=N^{>}_{G_{t''}}(v_i)$, where~$t'+1\leq t''\leq t$ achieves the
 maximum cardinality, is contained in $S^\star$. Let $S'$ be the smallest subset
 of~$S^\star$ such that in $\TG_{[1:t']}-S'$ the vertex~$v_{i'}$ is the largest
 reachable vertex from~$s$. By induction hypothesis, we have that $|S'|\ge
 |T[t',i']|$. From~\cref{lem:metaordpreinterv}\condRef{item:intervalsep4} it
 follows that $S'\cap S''=\emptyset$. Hence, because $S^\star$ is minimum, we can write
 $S^\star=S'\uplus S''$. %
 Hence, we have~
 $$  |S| = |S'|+|S''|\geq |T[t',i']| + |N^{>}_{G_{t''}}(v_i)|\geq \min_{S\in
 X_{t,i}} |S|= |T[t,i]|,$$ where the second inequality follows from the
 fact that~$T[t',i'] \cup N^{>}_{G_{t''}}(v_i)\in X_{t,i}$.
\end{proof}

Next we show how to use the derived poly\-no\-mi\-al-time algorithm as a basis
for a distance-to-triviality
parameterization~\cite{cai2003parameterized,guo2004structural}. For a
\utempinterval{} we introduce a parameter $\kappa$ that bounds how much the
compatible vertex orderings of two consecutive layers of a \utempinterval{}
differ. We use the \emph{Kendall tau} distance~\cite{kendall1938new} to measure
the similarity of vertex orderings. The Kendall tau distance~$K$ is a metric
that counts the number of pairwise disagreements between two total orderings; it is also known as ``bubble sort distance''. 
We call the parameter~$\kappa$ the \emph{shuffle number} of a \utempinterval{}
and define it as follows.
\begin{definition}[Shuffle Number]
Given a \utempinterval{}~$\TG=(V,\TE,\tau)$, its \emph{shuffle number} $\kappa$
is the smallest integer such that there are vertex orderings $<_V^1, <_V^2,
\ldots, <_V^\tau$ with the property that $<_V^t$ is compatible with layer $G_t$
for all~$t\in[\tau]$, and the orderings of any two consecutive layers have
Kendall tau distance at most $\kappa$, that is, for all $t\in[\tau-1]$ we have
that~$K(<_V^t,<_V^{t+1})\le\kappa$. We say that the vertex orderings $<_V^1, <_V^2, \ldots, <_V^\tau$ \emph{witness} the shuffle number of $\TG$.
\end{definition}
Clearly for \outempinterval{}s we have that $\kappa=0$ and it is easy to observe
(with the help of \cref{lemma:edge-layer}) that we get \NP-completeness for
$\kappa=1$.
However, if we consider the parameter combination $(\kappa+\tau)$ the problem becomes fixed-parameter tractable.

\begin{theorem}\label{thm:ktdist}
	Given the a \utempinterval{} and a vertex orderings that witness its
	shuffle number~$\kappa$,
	\nonstrproblem{} is fixed-parameter tractable when parameterized by $\kappa + \tau$,
	where $\tau$ is the maximum label.
\end{theorem}
 
 \begin{proof}
Let $\TG=(V,\TE,\tau)$ be a \utempinterval{} given as input together with vertex
orderings $<_V^1, <_V^2, \ldots, <_V^\tau$, and let $k$ be the size bound on the
separator. The algorithm proceeds as follows. We first ``mark'' all vertices~$u, v$ with the property that for some $t\in[\tau-1]$ we have that $u<_V^t v$ and $v<_V^{t+1}u$, that is, their relative order is flipped at some point in time. We also always mark $s$ and $z$. Let $M$ be the set of marked vertices. More formally, let $M$ be the largest subset of $V$ that contains $s$ and $z$ with the property that for all $u\in M\setminus \{s,z\}$ there is a $v\in M$ and a $t\in[\tau-1]$ such that either $u<_V^t v$ and $v<_V^{t+1}u$, or $v<_V^t u$ and $u<_V^{t+1}v$.

Note that we can compute~$M$ in polynomial time when given the vertex orderings
using bubble sort and we have that $|M|\le 2\cdot\kappa\cdot\tau+2$. If $M=V$, then we can solve the problem in the desired running time by trying out every possible separator. From now on we assume that~$M\neq V$.

Next, we define two partitions, one for the vertex set $M$ and one for the vertex set $V':=V\setminus M$. Intuitively, the partition of $V'$ describes which parts of the orderings 
stay the same over the whole lifetime of the temporal graph, or in other
words, which parts of the graph are order-preserving. 
The partition of $M$ describes which vertices lie between parts of the temporal graphs that are order-preserving.

We define a partition of the vertices in $M = M_1\uplus M_2 \uplus \ldots\uplus M_p$ as follows: Let $V=\{v_1, v_2, \ldots, v_n\}$ be the vertex ordering given by $<_V^1$ (that is, $v_i <_V^1 v_j$ if and only if $i<j$).
\begin{compactitem}
\item We have that $s\in M_1$ and $z\in M_p$.
\item If $v_i\in M$ and $v_{i+1}\in M$ for some $i\in[n-1]$, then $v_i\in M_j$ and $v_{i+1}\in M_j$ for some $j\in [p]$.
\item If $v_i\in M_j$ and $v_{i'}\in M_j$ with $i<i'$ for some $j\in[p]$, then for all $i<i^\star<i'$ we have that $v_{i^\star}\in M_j$.
\item For all $j\in[p]$ we have that $M_j\neq \emptyset$, and if $v_i$ in $M_j$
and $v_{i'}$ in $M_{j+1}$ for some $j\in[p-1]$, then we have that $i<i'$.
\end{compactitem}
Analogously, we define a partition of the remaining vertices $V'=V'_1\uplus V'_2 \uplus \ldots\uplus V'_q$ in the following way: 
\begin{compactitem}
\item If $v_i\in V'$ and $v_{i+1}\in V'$ for some $i\in[n-1]$, then $v_i\in V'_j$ and $v_{i+1}\in V'_j$ for some $j\in [q]$.
\item If $v_i\in V'_j$ and $v_{i'}\in V'_j$ with $i<i'$ for some $j\in[q]$, then for all $i<{i^\star}<i'$ we have that $v_{i^\star}\in V'_j$.
\item For all $j\in[q]$ we have that $V'_j\neq \emptyset$, and if $v_i$ in
$V'_j$ and $v_{i'}$ in $V'_{j+1}$ for some $j\in[q-1]$, then we have that
$i<i'$.
\end{compactitem}
We can easily compute both partitions by iterating over all vertices in $V$ in the order given by $<_V^1$ and checking whether a vertex is contained in $M$. It is also easy to check that $q\le p+1\le\kappa\cdot\tau+3\le n$, since for all~$1<j<p$ we have that $|M_j|\ge 2$.

Note that any vertex ordering $<_V^t$ with $t\in[\tau]$ defines the same
partitions. 

Now we are ready to construct a separator $S$. First we guess the set~$M_S:=S\cap M$. %
Then for each $1<j\le p$ we guess the earliest time $a_j$ a temporal path starting from~$s$ should be able to reach a vertex in the set $M_j$ in $\TG-S$ or we set $a_j:=\tau+1$ if no temporal path from~$s$ should be able to reach a vertex in~$M_j$ in $\TG-S$. 
For each~$1\le j<p$ we guess the earliest time $d_j> a_j$ a temporal path from $s$ should be able to reach a vertex in~$V'_j$  in $\TG-S$ or, in other words, leave the set $M_j$, or we set~$d_j:=\tau+1$ if no temporal path from~$s$ should be able to reach a vertex in $V'_j$ in $\TG-S$. 

Now we create the following instances of \nonstrproblem{} on \outempinterval{}s:
For each $j\in[q]$ we do the following: If $d_j<a_{j+1}$, then we create an
\outempinterval{} by taking the graph~$\TG_{[d_j:a_{j+1}-1]}[V'_j]$ and adding
two new vertices~$s_j$ and~$z_j$. We further add the time-edge $(\{s_j,u\},t)$ to
the temporal graph if $d_j\le t\le a_{j+1}-1$ and $(\{u',u\},t)\in\TE$ for some
$u'\in M_j\setminus M_S$. We add the edge $(\{z_j,u\},t)$ to the graph if $d_j\le
t\le a_{j+1}-1$ and $(\{u',u\},t)\in\TE$ for some $u'\in M_{j+1}\setminus M_S$.
We call the constructed graph $\TG_j$. Intuitively, we merge all vertices in $M_j\setminus M_S$ to a vertex $s_j$ and all vertices in $M_{j+1}\setminus M_S$ to a vertex $z_j$. It is easy to check that~$\TG_j$ is an \outempinterval{}. Now we solve the optimization variant of \nonstrproblem{} on $(\TG_j,s_j,z_j)$ using \cref{prop:intervalpoly}\footnote{Note that the corresponding algorithm can easily be modified to output a solution.}. Let~$S_j$ be the solution, that is, a minimum \nonstrsep{s_j,z_j} for $\TG_j$. If there is no valid solution or if $d_j\ge a_{j+1}$, then we set~$S_j=\emptyset$.

Finally, we set $S=\bigcup_{j\in[q]}S_j\cup M_S$. If $|S|\le k$ and there is no
\nonstrpath{s,z} in $\TG-S$, then we output \yes. Otherwise, we output \no.

It is easy to check that the algorithm runs in \FPT{}-time with respect to parameter $(\kappa+\tau)$. We next prove the correctness of the algorithm.

$(\Rightarrow)$: If the algorithm outputs \yes, then we face a \yes-instance. This is trivially true since the algorithm does a sanity check as a last step.

$(\Leftarrow)$: If we face a \yes-instance, then there is a \nonstrsep{s,z}~$S^\star$ with $|S^\star|\le k$ for $\TG$. We claim that in this case, our algorithm outputs \yes. Since we try out all possible sets $M_S$ we can assume that there is a branch of our algorithm where we have that $M_S=M\cap S^\star$. Similarly, we can assume that we are in a branch where all values $a_j$ and $d_j$ for $j\in[q]$ are ``correct'', that is, they are the largest numbers with the property that no vertex $v\in M_j$ is reachable from $s$ in $\TG-S^\star$ earlier than~$a_j$ and no vertex $u\in V'_j$ is reachable from $s$ in $\TG-S^\star$ earlier than~$d_j$. 

Then we can show that~$S=\bigcup_{j\in[q]}S_j\cup M_S$ is a \nonstrsep{s,z} and $|S|\le |S^\star|$: We first check that $S$ is a \nonstrsep{s,z}. Since $M\cap S^\star = M\cap S$ we know that for each part $M_j$ with $1<j<p$ we have that a temporal path from $s$ that arrives at a vertex in $M_j$ no earlier than~$a_j$ cannot arrive at a vertex in $V'_j$ earlier than $d_j$ in $\TG-S$. Furthermore, no temporal path from $s$ can arrive at a vertex in $V'_1$ earlier than $d_1$ in $\TG-S$ and no temporal path from $s$ that arrives at a vertex in $M_p$ at time $a_p$ or later can reach $z$ in $\TG-S$. The sets $S_j$ are chosen in a way that ensures that a temporal path from $s$ that does not arrive at any vertex in $V'_j$ earlier than~$d_j$ cannot reach a vertex in $M_{j+1}$ earlier than $a_{j+1}$ in $\TG-S_j$ and hence also in $\TG-S$. We can conclude that $S$ is a \nonstrsep{s,z} for~$\TG$. Now assume for contradiction that $|S|>|S^\star|$. Then there is a set~$S_j$ such that $|S_j|>|V'_j\cap S^\star|$. This is a contradiction to the fact that $S_j$ is a minimum \nonstrsep{s_j,z_j} for~$(\TG_j,s_j,z_j)$ since~$V'_j\cap S^\star$ is also a \nonstrsep{s_j,z_j} for $(\TG_j,s_j,z_j)$ since otherwise there would be a temporal path from $s$ that arrives at a vertex in $M_{j+1}$ earlier than $a_{j+1}$ in~$\TG-S^\star$. This completes the correctness proof.

\emph{Running time.} There are $2^{|M|}$ possible guesses for $M_S$ and then a total of~$\tau^{2(p-1)}$ possible guesses for the $a_i$ and $d_i$ values. The polynomial part of the running time is $q\cdot \ON(|V|^2\cdot \tau^2)$. Together with the bounds we know for~$q$,~$p$, and~$|M|$ we get a running time upper bound of $\ON((4\tau)^{\tau\cdot\kappa}\cdot (\kappa+\tau)\cdot |V|^2\cdot \tau^2)$.
 \end{proof}
 
We remark that is remains an open question how to compute the shuffle number of a given \utempinterval{} and a set of vertex orderings that witness the shuffle number. We conjecture that deciding whether a \utempinterval{} has shuffle number $\kappa=1$ is already \NP-hard.

\section{Conclusion}\label{sec:conclusion}

We studied \nonstrproblem{} on different temporal graph classes---with structural and temporal restrictions on temporal graph models.
We proved \nonstrproblem{} to remain~\NP{}-complete on the majority of the considered classes of restricted temporal graphs.
Polynomial-time solvability was achieved for temporal graphs where the underlying graph has bounded treewidth, on grounded temporal graphs, temporal graphs with many periods, and temporal graphs where each layer is a unit interval graph with respect to the same vertex ordering.

Our results exemplify that many notions of temporal graph classes that are currently considered in the literature do not impose useful restrictions on temporal graphs
when dealing with separation problems.
We introduced the class of~\outempinterval{}s which is more restrictive than just requiring the layers to fall into a specific graphs class. However, also this notion does not capture temporal aspects (that is, it is invariant under reordering of layers). We defined a distance measure for \utempinterval{} to \outempinterval{}, the \emph{shuffle number} of a \utempinterval{}, and showed that this is a useful restriction for \nonstrproblem{}.
Exploring further, more sophisticated structural restrictions of temporal graphs, whose definitions may rely on global properties and on temporal aspects, is of particular interest when asking for computationally tractable cases of~\nonstrproblem{}.

We further briefly discuss \textsc{Strict Temporal $(s,z)$-Separa\-tion}, the
main difference to \nonstrproblem{} being that we are looking for a
\emph{strict} temporal $(s,z)$-separator that removes all \emph{strict} temporal
$(s,z)$-paths from the input graph. 
A temporal path is \emph{strict} if the
time-edges of the path have strictly increasing time labels. 
In certain circumstances \textsc{Strict Temporal $(s,z)$-Separation} and
\nonstrproblem{} can behave very differently with respect to their computational
complexity~\cite{zschoche2017computational}, nevertheless we believe that most of our
results can be adapted to the strict case. 
More specifically, we believe that
the results presented in \cref{sec:layer} and \cref{sec:ug-classes} all carry
over, however the algorithms of course need suitable adjustments. 
For our results on temporal restrictions (\cref{sec:tem-classes}) it is easy to show that most of the polynomial-time solvable cases become \NP-hard in the strict case. This follows from the
fact that \textsc{Strict Temporal $(s,z)$-Separation} is \NP-complete even if
all layers are the same, or in other words, all edges appear either in all time
steps or never~\cite{zschoche2017computational}. We also believe that the
algorithm of \cref{sec:unitinterval} concerning \utempinterval{} can be adapted
to the strict case.
\bibliographystyle{abbrvnat}
\bibliography{library}
\clearpage
\appendix
\appendixproof
\end{document}